\documentclass{amsart}

\usepackage{amsmath,amssymb}
\usepackage{amsthm}
\usepackage[alphabetic]{amsrefs}
\usepackage{indentfirst}
\usepackage{mathrsfs}
\usepackage{hyperref}
\usepackage{fourier}
\usepackage[margin=1.5in]{geometry}
\usepackage{enumitem}
\usepackage{mathtools}
\usepackage{xcolor}
\theoremstyle{plain}
\newtheorem{theorem}{Theorem}[section]
\newtheorem*{informaltheorem}{Main results}
\newtheorem{lemma}[theorem]{Lemma}

\newtheorem{prop}[theorem]{Proposition}

\theoremstyle{definition}

\newtheorem{assump}[theorem]{Assumption}

\theoremstyle{remark}
\newtheorem{remark}[theorem]{Remark}

\DeclareMathOperator{\tr}{tr}

\DeclareMathOperator{\Ent}{Ent}

\DeclareMathOperator{\diag}{diag}

\newcommand{\dd}{\,\mathrm{d}}

\IfFileExists{mathabx.sty}%
  {\DeclareFontFamily{U}{mathx}{\hyphenchar\font45}%
   \DeclareFontShape{U}{mathx}{m}{n}{<->mathx10}{}%
   \DeclareSymbolFont{mathx}{U}{mathx}{m}{n}%
   \DeclareFontSubstitution{U}{mathx}{m}{n}%
   \DeclareMathAccent{\widebar}{0}{mathx}{"73}%
}{%
  \PackageWarning{mathabx}{%
    Package mathabx not available, therefore\MessageBreak substituting
    widebar with overline\MessageBreak }%
  \newcommand{\widebar}[1]{\overline{#1}}%
}

\newcommand{\mc}[1]{\mathcal{#1}}

\newcommand{\norm}[1]{\lVert#1\rVert}

\renewcommand{\Re}{\mathfrak{Re}}
\renewcommand{\Im}{\mathfrak{Im}}

\newcommand{\R}{\mathbb{R}}
\newcommand{\E}{\mathbb{E}}
\newcommand{\cI}{\mathcal{I}}
\newcommand{\cJ}{\mathcal{J}}

\newcommand{\cR}{\mathcal{R}}
\newcommand{\ip}[2]{\left\langle #1,#2\right\rangle}

\title[Log-Sobolev inequalities for boundary-driven chains]{Log-Sobolev inequalities for boundary-driven anharmonic chains}
\author{Jianfeng Lu}
\address{Mathematics Department, Duke University, Box 90320, Durham, NC 27705 USA.}
\email{jianfeng@math.duke.edu}
\thanks{The work of J.L. is supported in part by the National Science Foundation under grant DMS-2309378. We thank Bowen Li, Jonathan Mattingly, Felix Otto, and Yuliang Wang for helpful discussions.}
\date{\today}
\hypersetup{
  pdftitle={Log-Sobolev inequalities for boundary-driven anharmonic chains},
  pdfauthor={Jianfeng Lu}
}

\begin{document}
\begin{abstract}
We study the non-equilibrium steady state of a weakly anharmonic chain of $N$ oscillators driven at its boundary by Langevin thermostats at unequal temperatures. 
Under a perturbative weak-anharmonicity condition, we prove a full-gradient logarithmic Sobolev inequality whose constant is independent of the chain length $N$.  
For homogeneous pinned chains, an additional quantitative regularity assumption yields a boundary space-time logarithmic Sobolev inequality and relative-entropy decay on the same $O(N^3)$ relaxation time scale as the harmonic chain.
The proof extracts a finite-dimensional Gaussian component from the boundary noise and compares conditional terminal-state laws by a change of variables.  The estimates are uniform over bounded positive temperatures and require no near-equilibrium assumption on their difference.
\end{abstract}

\maketitle

\section{Introduction}\label{sec:intro}

Non-equilibrium steady states (NESS) for systems coupled to thermostats provide
microscopic models of stationary heat transport.  Compared with equilibrium
Gibbs states, their analysis faces several challenges: the invariant density is
usually not explicit, the dynamics is non-reversible, and entropy decreases
only through the variables in contact with the boundary thermostats.  We study
functional inequalities for a prototypical example, a boundary-driven
oscillator chain.  Fix $N\ge2$, write
$z=(q,p)\in\R^N\times\R^N$, and consider the pinned, weakly anharmonic chain
with Hamiltonian
\[
   H_N(q,p)=\tfrac12|p|^2+V_N(q),\qquad
   V_N(q)=V_N^0(q)+W_N(q),
\]
whose quadratic part is a homogeneous chain of on-site and nearest-neighbor
springs,
\[
   V_N^0(q)=\frac\nu2\sum_{i=1}^N q_i^2
      +\frac\kappa2\sum_{i=0}^N(q_{i+1}-q_i)^2,
   \qquad q_0=q_{N+1}=0,\quad \nu,\kappa>0.
\]
Thus each oscillator is pinned to its site with strength $\nu$ and coupled,
with strength $\kappa$, to its neighbors and, at the two ends, to fixed walls.
We later allow a broader class of harmonic reference chains.
The perturbation $W_N$ is a weak anharmonic potential, to be specified later.
The boundary sites $\partial=\{1,N\}$ are coupled to Langevin
thermostats at temperatures $T_1,T_N$, each with friction coefficient
$\gamma>0$, so that the resulting dynamics is the
boundary-driven hypoelliptic diffusion
\begin{equation}\label{eq:SDE}
\begin{aligned}
  \dd q_i(t)&=p_i(t)\,\dd t, &&1\le i\le N,\\
  \dd p_1(t)&=-\partial_{q_1}V_N(q_t)\,\dd t-\gamma p_1(t)\,\dd t
     +\sqrt{2\gamma T_1}\,\dd W_{\mathrm{L},t},\\
  \dd p_i(t)&=-\partial_{q_i}V_N(q_t)\,\dd t, &&2\le i\le N-1,\\
  \dd p_N(t)&=-\partial_{q_N}V_N(q_t)\,\dd t-\gamma p_N(t)\,\dd t
     +\sqrt{2\gamma T_N}\,\dd W_{\mathrm{R},t},
\end{aligned}
\end{equation}
where $W_{\mathrm{L},t},W_{\mathrm{R},t}$ are independent standard Brownian motions.
When $T_1\ne T_N$, the invariant measure $\pi_N$, when it exists, is the
non-equilibrium steady state.

Although $\pi_N$ is not explicit, entropy inequalities allow us to study both
its static properties and the approach to stationarity.  Throughout, $f$ denotes a probability density with respect to $\pi_N$, we
write $\Ent_{\pi_N}(f)$ for its relative entropy, and $\mc P_t^*f$ for the
density obtained by evolving the law $f\pi_N$ for time $t$ under \eqref{eq:SDE}.  We focus on two
questions:

\begin{enumerate}[label=(\Alph*),leftmargin=*]
\item \emph{How concentrated is $\pi_N$?}  A dimension-free full-gradient LSI,
which controls the relative entropy of $f$ by its full-gradient Fisher
information $\cI_{\mathrm{full}}(f)$, defined in \eqref{eq:full-fisher}, with a
constant of size $O(1)$, gives uniform fluctuation bounds for NESS observables.
This Fisher information is the Dirichlet energy of $\sqrt f$ in all $2N$
phase-space directions.
\item \emph{How fast does the chain reach $\pi_N$?}  A boundary space-time LSI,
which controls the relative entropy of $f$ by its boundary Fisher information
$\cI_\partial(f)$, defined in \eqref{eq:boundary-fisher}, integrated along the
flow, gives relaxation from arbitrary finite-entropy initial densities.  This
Fisher information involves only derivatives in the two boundary momenta.
\end{enumerate}

We answer both questions in a weakly anharmonic perturbative regime.  The
precise smallness assumptions appear later in
Assumptions~\ref{ass:admissible-potential} and~\ref{ass:weak}.  The
dimension-free full-gradient LSI holds for a slightly broader class of harmonic
reference chains, specified in Section~\ref{sec:model}.

\begin{informaltheorem}[informal]
Let $\pi_N$ be the NESS of the boundary-driven chain and define
\[
   \delta_N=\sup_{q\in\R^N}\norm{\nabla^2W_N(q)}_{\mathrm{op}},
   \qquad
   \Lambda_N=
   \sup_{q\in\R^N}\sup_{|h|=1}
   \norm{D(\nabla^2W_N)(q)[h]}_{\mathrm{HS}},
\]
where $\norm{\cdot}_{\mathrm{op}}$ and $\norm{\cdot}_{\mathrm{HS}}$ are the operator and Hilbert--Schmidt norms on $N\times N$ matrices.
Then the following statements hold with constants independent of $N$.
\begin{enumerate}[label=\textup{(\Alph*)},leftmargin=*]
\item \emph{Dimension-free full-gradient LSI.}  Under the admissibility
condition of Assumption~\ref{ass:admissible-potential}, every density $f$ with
respect to $\pi_N$ obeys
   \[
      \Ent_{\pi_N}(f)\le C_{\mathrm{LSI}}\,\cI_{\mathrm{full}}(f).
   \]
\item \emph{Quantitative boundary space-time LSI.}  For the homogeneous pinned
chain, if $N^3\delta_N$ is sufficiently small and
$\Lambda_N<\infty$, there is a window
\[
   T_\ast\le C N^3\bigl[1+\log(2+(N^3\Lambda_N)^2)\bigr]
\]
such that
\[
   \Ent_{\pi_N}(f)
      \le 2\int_0^{T_\ast}
            \cI_\partial(\mc P_t^*f)\,\dd t.
\]
If $\sup_N N^3\Lambda_N<\infty$, then $T_\ast\le\tau_*N^3$ and
\[
   \Ent_{\pi_N}(\mc P_t^*f)
      \le 2e^{-c_*t/N^3}\Ent_{\pi_N}(f),
\]
where $\tau_*$ and $c_*$ are independent of $N$.
\end{enumerate}
\end{informaltheorem}

Result~(A) is a static concentration estimate for the NESS, uniform in the
chain length.  Result~(B) is dynamical: it converts dissipation through only
the two boundary momenta into relative-entropy decay from every finite-entropy
initial density.  When the additional regularity bound is uniform in $N$, the
decay occurs at the rate $cN^{-3}$ of the harmonic chain.

Although the two estimates are tied to the same entropy-production identity,
they are not redundant: neither implies the other.  The dimension-free
full-gradient LSI does not by itself yield the boundary space-time LSI\@.
Entropy production dissipates only the boundary Fisher information,
$\frac{\dd}{\dd t}\Ent_{\pi_N}(\mc P_t^*f)
=-\cI_\partial(\mc P_t^*f)$, whereas the full-gradient LSI uses
$\cI_{\mathrm{full}}$.  Conversely, the boundary space-time LSI is a time-integrated
boundary statement, not an instantaneous full-gradient LSI\@.  An instantaneous
boundary-only LSI is impossible: a nonconstant normalized density depending
only on positions has zero boundary Fisher information but positive entropy.
The time integration is what allows Hamiltonian transport to carry dissipation
from the boundary thermostats to the bulk.

The two LSIs are nevertheless connected.  Integrating the
entropy-production identity and using
$\Ent_{\pi_N}(\mc P_t^*f)\to0$ gives, with the quantitative time horizon
$T_\ast$ of Theorem~\ref{thm:main-stlsi},
\begin{equation}\label{eq:intro-entropy-chain}
  \underbrace{\tfrac12\Ent_{\pi_N}(f)
     \ \le\ \int_0^{T_\ast}\!\cI_\partial(\mc P_t^*f)\,\dd t}_{\text{boundary space-time LSI, (B)}}
  \ \le\ \int_0^\infty\!\cI_\partial(\mc P_t^*f)\,\dd t
  \ =\ \Ent_{\pi_N}(f)
  \ \le\
  \underbrace{C_{\mathrm{LSI}}\,\cI_{\mathrm{full}}(f)}_{\substack{\text{dimension-free}\\\text{full-gradient LSI, (A)}}}.
\end{equation}
Both results therefore bound the same entropy.  The boundary space-time LSI
uses boundary Fisher information over a quantitative finite window, while the
dimension-free full-gradient LSI uses the instantaneous full-gradient Fisher
information.  The middle equality in \eqref{eq:intro-entropy-chain} identifies
the total boundary Fisher information exactly with the entropy.  When
$\sup_N N^3\Lambda_N<\infty$, the finite window satisfies
$T_\ast\le\tau_*N^3$.

The proofs use different consequences of the boundary noise.  For Result~(A),
a Gaussian logarithmic Sobolev inequality is transferred from the noise to the
finite-time laws through a uniform estimate on the linearized propagator, and then
passed to the NESS.  For Result~(B), controllability of the harmonic chain and conditioning
reduce the comparison of transition kernels to a finite-dimensional change of
variables, which preserves the $N^3$ time scale under weak
anharmonicity.

The rest of the paper is organized as follows.  Section~\ref{sec:model} fixes
the model and states the precise assumptions and main results.
Section~\ref{sec:inputs} establishes estimates for the harmonic chain on the
$N^3$ relaxation and controllability scale.
Section~\ref{sec:lsi} establishes the uniform dual boundary estimate and proves
the dimension-free full-gradient LSI, while Section~\ref{sec:stlsi} proves the
boundary space-time LSI.

\subsection*{Related work}
The mathematical study of boundary-driven oscillator chains begins with the
model of a harmonic chain.  Rieder--Lebowitz--Lieb computed the Gaussian stationary state
and the heat current of a finite harmonic chain coupled to boundary reservoirs
\cite{RiederLebowitzLieb}.  Further early work includes
\cites{CasherLebowitz,SpohnLebowitz}; see
\cites{BonettoLebowitzReyBellet,LepriLiviPoliti} for broader background.

For anharmonic chains, the foundational rigorous results mainly concern
existence, uniqueness, and convergence at fixed system size.
Eckmann--Pillet--Rey-Bellet prove existence of a stationary state for arbitrary
temperature differences and uniqueness and mixing near equilibrium
\cite{EckmannPilletReyBellet}.  Their companion work extends uniqueness and
convergence to arbitrary temperature differences and relates entropy
production to the mean energy flux
\cite{EckmannPilletReyBelletEntropy}.  Eckmann--Hairer extend existence and
uniqueness to strongly anharmonic potentials with essentially arbitrary growth,
using a strengthened H\"ormander theory on unbounded domains
\cite{EckmannHairer}.  Rey-Bellet and Thomas prove exponential convergence for
chains coupled to Hamiltonian reservoirs \cite{ReyBelletThomas}, while Carmona
obtains exponential convergence for related stochastic heat-bath models
\cite{Carmona}.  Cuneo, Eckmann, Hairer, and Rey-Bellet establish existence,
uniqueness, and exponential convergence for more general oscillator networks
under controllability and assumptions on the potentials at infinity
\cite{CuneoEHRB}.  Taken together, these works establish fixed-$N$ ergodicity
by combining propagation of the boundary noise, deterministic control, and
Lyapunov--Harris arguments.  Our study is more on the quantitative side, focusing on the dependence of
functional inequalities and relaxation estimates on the chain length.

System-size dependence is understood most sharply for harmonic chains.
For the homogeneous pinned chain, Becker--Menegaki determine the sharp
$N^{-3}$ scaling of the spectral gap \cite{BeckerMenegaki}.  Adapting
Baudoin's generalized Bakry--\'Emery method \cite{Baudoin}, Menegaki proves
quantitative Wasserstein and relative-entropy convergence for weakly
anharmonic chains \cite{Menegaki}.  The work assumes that the sum
of the pinning and interaction Hessian bounds is $O(N^{-9/2})$ and gives a
Wasserstein estimate with rate $cN^{-3}$ and prefactor $O(N^{3/2})$.  Its
twisted-gradient LSI has constant $O(N^3)$
\cite{Menegaki}*{Proposition~1.5}.  The twisted gradient is comparable to the
Euclidean full gradient only up to a factor $O(N^3)$
\cite{Menegaki}*{equation~(3.2)}, since the twisting matrix has norm $O(N^3)$
while its inverse has norm $O(1)$.  Combining the two therefore yields an LSI
for the Euclidean full gradient with constant $O(N^6)$.
The full-gradient LSI proved here instead has a constant independent of $N$.
Its proof passes Gross's Gaussian LSI from finite-dimensional polygonal
approximations of the boundary Brownian motion to the transition laws and then
to the NESS, using an estimate for the adjoint of the variational propagator measured in the
norm dual to the energy norm of the harmonic chain.  For the homogeneous chain,
the full-gradient LSI applies when $N^3\delta_N$ is sufficiently small.

The closest general results concern functional inequalities for non-equilibrium
diffusions with non-explicit invariant measures.  Monmarch\'e and Wang develop
two approaches to logarithmic Sobolev inequalities for such measures: a
Holley--Stroock and Aida--Shigekawa perturbation, in which the non-explicit
correction is controlled by stochastic control, and a method combining
Wasserstein contraction with hypercontractivity \cite{MonmarcheWang}.  Huang,
Kopfer, Monmarch\'e, and Ren prove relative-entropy and Wasserstein ergodicity
for kinetic SDEs whose drift is dissipative outside a compact set, by deriving
hypercontractivity from hyperboundedness and $L^2$-ergodicity
\cite{HuangKopferMonmarcheRen}.  These results do not apply to the present
chain because the boundary noise is degenerate on the bulk momenta.

Finally, Li and Lu develop a space-time logarithmic Sobolev inequality and hypocoercive hypercontractivity for equilibrium underdamped Langevin dynamics
\cite{LiLu2026}, using the 
modified-entropy method with an optimal-transport corrector \cite{Lu2026}.  At
the $L^2$ level, Albritton--Armstrong--Mourrat--Novack develop a
variational theory and functional inequalities of Poincar\'e 
type for the kinetic Fokker--Planck equation \cite{AAMN}, while Cao--Lu--Wang
obtain explicit $L^2$ hypocoercive rates \cite{CaoLuWang2023}.  See
\cite{Villani} for general background on hypocoercivity.  These works connect
space-time or modified functional inequalities with hypocoercive relaxation in
equilibrium.  Our boundary space-time LSI plays an analogous role for a
different degeneracy: dissipation acts only on the two boundary momenta and the
invariant measure is a non-explicit NESS. 

\subsection*{Use of AI tools}
LLM-based assistants were used in the preparation of this manuscript for
drafting, language editing, and consistency checks.  The author is responsible
for independently checking every statement and proof and for the correctness
and integrity of the final manuscript.

\section{The model and main results}\label{sec:model}

\subsection{Boundary-driven oscillator chains}

We study a nearest-neighbor chain of $N$ oscillators
$(q,p)\in\R^N\times\R^N$, with Hamiltonian
$H_N(q,p)=\tfrac12|p|^2+V_N(q)$, coupled to Langevin thermostats at its boundary
sites $\partial=\{1,N\}$.  The
dynamics is the boundary-driven diffusion
\eqref{eq:SDE} introduced in Section~\ref{sec:intro}.  To write it more compactly, set
\begin{equation}\label{eq:sigma}
  \Gamma_\partial=\gamma(e_1e_1^{\top}+e_Ne_N^{\top}),
  \qquad
  B_{\partial,N}=
  \begin{pmatrix}
    0_{N\times1} & 0_{N\times1} \\
    \sqrt{2\gamma T_1}\,e_1 & \sqrt{2\gamma T_N}\,e_N
  \end{pmatrix}\in\R^{2N\times2},
\end{equation}
where $e_i$ are the coordinate vectors in $\R^N$. For $z=(q,p)$, define
\[
   b(z):=\binom{p}{-\nabla V_N(q)-\Gamma_\partial p}.
\]
Then the SDE \eqref{eq:SDE} can be written as
\begin{equation}\label{eq:SDE-vector}
  \dd Z_t=b(Z_t)\dd t
  +B_{\partial,N}\dd W_t,
\end{equation}
where
$W_t=(W_{\mathrm{L},t},W_{\mathrm{R},t})$ collects the two boundary noises into a
standard two-dimensional Brownian motion.
The associated generator is the boundary-driven hypoelliptic operator
\[
   \mc L_N=p\cdot\nabla_q-\nabla V_N(q)\cdot\nabla_p
      +\gamma\Bigl(T_1\partial_{p_1}^2-p_1\partial_{p_1}
         +T_N\partial_{p_N}^2-p_N\partial_{p_N}\Bigr).
\]
The noise and damping act \emph{only} at the boundary sites, and the temperatures $T_1,T_N$ may be unequal and satisfy
\begin{equation}\label{eq:T-bounds}
   0<\underline T\le T_1,T_N\le\overline T<\infty.
\end{equation}

We write the potential as a perturbation of a harmonic reference,
$V_N=V_N^0+W_N$.  Throughout the paper, $W_N\in C^\infty(\R^N)$ and is
nearest-neighbor at the Hessian level:
$\partial_{q_iq_j}^2W_N\equiv0$ whenever $|i-j|>1$.  Put
\[
   \delta_N:=\norm{\nabla^2W_N}_\infty,
   \qquad
   \Lambda_N:=
   \sup_{q\in\R^N}\sup_{|h|=1}
      \norm{D(\nabla^2W_N)(q)[h]}_{\mathrm{HS}}.
\]
Here $D(\nabla^2W_N)(q)[h]$ denotes the directional derivative of the Hessian at $q$ in the direction $h$; its $(i,j)$ entry is
$\sum_{k=1}^N\partial_{q_iq_jq_k}^3W_N(q)h_k$.  Moreover,
$\norm{\nabla^2W_N}_\infty:=
\sup_q\norm{\nabla^2W_N(q)}_{\mathrm{op}}$ uses the Euclidean operator norm, while
$\norm{\cdot}_{\mathrm{HS}}$ is the Euclidean Hilbert--Schmidt norm on $N\times N$
matrices.  The quantity $\delta_N$ is the perturbative Hessian
size used in both main results.  The Hessian-Lipschitz quantity $\Lambda_N$ is
used only for the boundary space-time LSI\@.  The affine part of $W_N$ is left
unconstrained, since differences of trajectories and the linearized equations depend
on $W_N$ only through its Hessian.  The reference is quadratic,
$V_N^0(q)=\frac{1}{2}q^{\top}K_N^0q$, and has the same nearest-neighbor
structure as the chain: its stiffness matrix is
symmetric tridiagonal and uniformly elliptic,
\begin{equation}\label{eq:K-elliptic}
   m_0\,I_N\le K_N^0 = \nabla^2V_N^0\le M_0\,I_N,\qquad m_0,M_0>0\ \text{independent of }N.
\end{equation}
For the homogeneous chain of Section~\ref{sec:intro}, with on-site strength
$\nu$ and coupling $\kappa$, $K_N^0$ is symmetric tridiagonal with diagonal
entries $\nu+2\kappa$ and off-diagonal entries $-\kappa$.  Thus
\eqref{eq:K-elliptic} holds with $m_0=\nu$ and $M_0=\nu+4\kappa$.
We use the associated energy inner product, its norm, and the Euclidean dual
norm,
\begin{equation}\label{eq:energy-dual-norms}
\begin{aligned}
   \langle z,z'\rangle_E
      &=q^{\top}K_N^0q'+p^{\top}p',
      \qquad z=(q,p),\quad z'=(q',p'),\\
   \norm{z}_E^2&=q^{\top}K_N^0q+|p|^2,\qquad z=(q,p),\\
   \norm{w}_{E,*}^2&=w_q^{\top}(K_N^0)^{-1}w_q+|w_p|^2,
      \qquad w=(w_q,w_p).
\end{aligned}
\end{equation}
By the ellipticity bounds \eqref{eq:K-elliptic}, the energy norm and its dual
are uniformly equivalent to the Euclidean norm, with constants independent of
$N$: for all $z,w\in\R^{2N}$,
\begin{equation}\label{eq:norm-equivalence}
   \min\{m_0,1\}\,|z|^2\le\norm{z}_E^2\le\max\{M_0,1\}\,|z|^2,
   \qquad
   \frac{|w|^2}{\max\{M_0,1\}}\le\norm{w}_{E,*}^2\le\max\{1/m_0,1\}\,|w|^2 .
\end{equation}
In particular $|w|^2\le\max\{M_0,1\}\norm{w}_{E,*}^2$.
The energy norm induces the extended quadratic Wasserstein distance
\[
   \mathsf{W}_E(\mu_1,\mu_2)^2
      :=\inf_{\lambda\in\Pi(\mu_1,\mu_2)}
        \int_{\R^{2N}\times\R^{2N}} \norm{z-z'}_E^2\,\lambda(\dd z,\dd z'),
\]
for probability measures $\mu_1,\mu_2$ on $\R^{2N}$, where
$\Pi(\mu_1,\mu_2)$ is the set of their couplings.  The distance is understood
to be $+\infty$ when the infimum is infinite.

For each fixed $N$, under the admissibility hypothesis below, the dynamics
\eqref{eq:SDE} is ergodic by Proposition~\ref{prop:fixedN-mixing} and has a
unique invariant probability measure with smooth strictly positive density.  We
denote this invariant measure by $\pi_N$.

Let $\mc P_t=e^{t\mc L_N}$ be the Markov semigroup acting on observables,
\[
   \mc P_tg(z)=\E_z g(Z_t).
\]
Denote by $\mc P_t^*$ the adjoint of $\mc P_t$ on densities with respect to
$\pi_N$:
\[
   \int (\mc P_t\varphi)\,f\dd\pi_N=\int\varphi\,(\mc P_t^* f)\dd\pi_N.
\]
If a law has density $f$ with respect to $\pi_N$, then $\mc P_t^* f$ is the
density of its evolution by the SDE at time $t$.  Since the dynamics
\eqref{eq:SDE} is non-reversible, $\mc P_t^*\ne\mc P_t$ in general.

Define the boundary carr\'e du champ by
\[
   \Gamma(f)=\gamma\sum_{i\in\partial}T_i\,|\partial_{p_i}f|^2.
\]
Let $\mc E_\partial$ be the closure in $L^2(\pi_N)$ of the quadratic form
\[
   \mc E_\partial(u,u)
      :=\int\Gamma(u)\dd\pi_N
      =\gamma\sum_{i\in\partial}T_i
         \int|\partial_{p_i}u|^2\dd\pi_N,
   \qquad u\in C_c^\infty(\R^{2N}).
\]
It is closable because $\pi_N$ has a smooth strictly positive density: on
every compact set its weight is bounded above and below, so a sequence
converging to zero in $L^2(\pi_N)$ whose boundary gradients are Cauchy has
zero distributional boundary-gradient limit.  The usual cutoff argument also
shows
that every locally Sobolev function in $L^2(\pi_N)$ with finite displayed
boundary energy belongs to the form domain.

For any density $f$ with respect to $\pi_N$, define its boundary Fisher
information by
\begin{equation}\label{eq:boundary-fisher}
   \cI_\partial(f):=4\,\mc E_\partial(\sqrt f,\sqrt f),
\end{equation}
with value $+\infty$ when $\sqrt f$ is outside the form domain.  If $f$ is
smooth and strictly positive, the chain rule gives the familiar formula
\[
   \cI_\partial(f)=\int\frac{\Gamma(f)}{f}\dd\pi_N
      =4\gamma\sum_{i\in\partial}T_i
         \int |\partial_{p_i}\sqrt f|^2\dd\pi_N.
\]
The definition through the closed quadratic form is needed only to state the
entropy identity for nonsmooth data.

For a probability measure $\mu$ and a nonnegative integrable function $g$, set
\[
   \Ent_\mu(g)
      :=\int g\log\!\left(\frac{g}{\int g\dd\mu}\right)\dd\mu,
\]
with value zero when $g=0$ $\mu$-almost everywhere.  For probability measures
$\mu_1,\mu_2$, write $\mathsf{KL}(\mu_1\mid\mu_2)$ for their relative entropy,
with value $+\infty$ when $\mu_1\not\ll\mu_2$.  In particular, if $f$ is a
density with respect to $\pi_N$, then
\[
   \Ent_{\pi_N}(f)=\int f\log f\dd\pi_N
      =\mathsf{KL}(f\pi_N\mid\pi_N).
\]
For smooth positive data for which differentiation and integration by parts
are justified, set $h_t=\mc P_t^*f$.  Invariance of $\pi_N$ and the diffusion
chain rule give
\begin{equation}\label{eq:entropy-id}
\begin{aligned}
   \frac{\dd}{\dd t}\Ent_{\pi_N}(h_t)
      &=\int(1+\log h_t)\,\partial_t h_t\dd\pi_N \\
      &=-\int\frac{\Gamma(h_t)}{h_t}\dd\pi_N
       =-\cI_\partial(h_t).
\end{aligned}
\end{equation}
The next lemma records the corresponding exact integrated identity for
arbitrary finite-entropy data; we omit the proof for the standard result, while referring to Fontbona and Jourdain
\cite{FontbonaJourdain}. 

\begin{lemma}[Exact integrated entropy dissipation]\label{lem:integrated-entropy}
If $f$ is a density with $\Ent_{\pi_N}(f)<\infty$, then
$t\mapsto\Ent_{\pi_N}(\mc P_t^*f)$ is nonincreasing and, for
$0\le s<t<\infty$,
\begin{equation}\label{eq:integrated-entropy}
   \Ent_{\pi_N}(\mc P_s^*f)-\Ent_{\pi_N}(\mc P_t^*f)
      =\int_s^t\cI_\partial(\mc P_r^*f)\,\dd r.
\end{equation}
The integrand is understood through the closed boundary form in \eqref{eq:boundary-fisher}.
\end{lemma}

For the static LSI, let $\mc E_{\mathrm{full}}$ be the closure in $L^2(\pi_N)$ of
\[
   \mc E_{\mathrm{full}}(u,u):=\int|\nabla u|^2\dd\pi_N,
   \qquad u\in C_c^\infty(\R^{2N}),
\]
where $\nabla=(\nabla_q,\nabla_p)$ is the Euclidean phase-space gradient.
Define the full Fisher information of a density $f$ by
\begin{equation}\label{eq:full-fisher}
   \cI_{\mathrm{full}}(f):=4\,\mc E_{\mathrm{full}}(\sqrt f,\sqrt f),
\end{equation}
with value $+\infty$ outside the form domain.  For smooth positive $f$, the
Sobolev chain rule gives
\[
   \cI_{\mathrm{full}}(f)
      =\int\frac{|\nabla f|^2}{f}\dd\pi_N
      =4\int|\nabla\sqrt f|^2\dd\pi_N.
\]
As for $\mc E_\partial$, the smooth positive density of $\pi_N$ makes the
pre-form closable, and the usual cutoff argument gives its weighted Sobolev
characterization.

\subsection{The dimension-free full-gradient LSI}

We first state the dimension-free full-gradient LSI under the following
assumption.

\begin{assump}[Admissible perturbation]\label{ass:admissible-potential}
The harmonic reference $K_N^0$ is symmetric tridiagonal and obeys
\eqref{eq:K-elliptic}, and the perturbation $W_N$ is \emph{admissible}:
\begin{enumerate}[label=(\roman*)]
\item the perturbed bonds remain nondegenerate:
   \begin{equation}\label{eq:bond-nondegenerate}
      |\partial_{q_iq_{i+1}}^2V_N(q)|
      \ge\frac12\min_{1\le j<N}|(K_N^0)_{j,j+1}|>0,
      \qquad 1\le i<N,\quad q\in\R^N;
   \end{equation}
\item the perturbation is sufficiently small:
   \[
      \Theta_N:=\norm{\nabla^2W_N}_\infty
      \max\{\norm{G_N},m_0^{-1}\}\le\theta_0,
   \]
   where $G_N$ is defined in \eqref{eq:GN} and $\theta_0>0$ is sufficiently small,
   depending only on $m_0,M_0,\gamma,\underline T,\overline T$.
\end{enumerate}
The first condition guarantees that no adjacent coupling vanishes.  This
allows the boundary noise to propagate to all sites.  The second condition is a smallness assumption and carries the $N^3$ scale.  By
Lemma~\ref{lem:gramian-bound}, on the homogeneous chain,
\[
   \norm{G_N}\le CN^3
   \quad \Rightarrow \quad
   \Theta_N\le CN^3\norm{\nabla^2W_N}_\infty,
\]
which gives the natural $N^3$ scaling for the homogeneous chain.
\end{assump}

Under this admissibility hypothesis, the boundary noise propagates to every site and yields ergodicity.  We record the fixed-$N$ statement here; the proof is standard by a Lyapunov--Harris argument, see e.g., \cite{HairerMattinglyHarris}.  Smoothness and strict positivity of the invariant density come from the H\"ormander condition, which \eqref{eq:bond-nondegenerate} guarantees; see \cite{CuneoEHRB}.

\begin{prop}[Fixed-$N$ ergodicity]\label{prop:fixedN-mixing}
Assume $W_N$ is admissible.  The semigroup has a unique invariant probability
measure $\pi_N$ with smooth strictly positive density and finite second moment
in the energy norm.  For every $z\in\R^{2N}$,
\begin{equation}\label{eq:pointwise-TV}
   \norm{\mc P_t(z,\cdot)-\pi_N}_{\mathrm{TV}}\longrightarrow0.
\end{equation}
\end{prop}

By \eqref{eq:pointwise-TV} and dominated convergence, $\mc P_t^*f\to1$ in
$L^1(\pi_N)$ for every density $f$, and a standard truncation argument upgrades
this to $\Ent_{\pi_N}(\mc P_t^*f)\to0$ whenever $\Ent_{\pi_N}(f)<\infty$.
Letting $t\to\infty$ in Lemma~\ref{lem:integrated-entropy} then gives
\begin{equation}\label{eq:entropy-id-infinite-adjoint}
   \Ent_{\pi_N}(f)=\int_0^\infty
      \cI_\partial(\mc P_t^*f)\,\dd t ,
\end{equation}
the infinite-time equality in \eqref{eq:intro-entropy-chain}; the proof below
uses the finite-time identity directly.  Note that no formula for the generator
of $\mc P_t^*$ is needed in the argument, which is in general not available. 

With the invariant measure fixed, we can state the first main estimate.  It is
a static inequality for the NESS: relative entropy is controlled by the
full-gradient Fisher information with a constant independent of the chain
length.

\begin{theorem}[Dimension-free full-gradient LSI]\label{thm:main-lsi}
Under Assumption~\ref{ass:admissible-potential}, there is a constant
$C_{\mathrm{LSI}}$, depending only on $m_0,M_0$, $\gamma$, and the temperature
bounds in \eqref{eq:T-bounds}, such that for each $N$
\begin{equation}\label{eq:main-lsi}
   \Ent_{\pi_N}(f)\le C_{\mathrm{LSI}}\,\cI_{\mathrm{full}}(f)
\end{equation}
for every density $f$ with respect to $\pi_N$.
\end{theorem}

We outline the proof, which is given in Section~\ref{sec:lsi}.  Because the
invariant measure $\pi_N$ is not explicit, we first prove the same inequality
for the law of the chain at a finite time, starting from a fixed state.  The
constant is uniform in the starting state, the elapsed time, and the chain
length.  As time tends to infinity, these laws converge to $\pi_N$ by
Proposition~\ref{prop:fixedN-mixing}, and the inequality passes to the limit.

The advantage of working at finite time is that the final state is determined
by the two boundary Brownian motions.  We approximate these motions by paths
depending on finitely many Gaussian variables and apply the Gaussian
logarithmic Sobolev inequality \cite{Gross1975}.  This reduces the proof to
controlling how strongly the final state responds to changes in the boundary
noise.  That response is governed by the linearization of the dynamics along
the trajectory, and the perturbation enters only through its Hessian.

It remains to show that the accumulated response to the two boundary noises is
bounded uniformly in both time and chain length.  For the harmonic chain, an
energy identity gives this bound directly: the conservative part of the
dynamics cancels, leaving only the two boundary terms
(Proposition~\ref{prop:dual-boundary-energy}).  The same estimate remains valid
for the weakly anharmonic chain because the perturbation is small in the sense
of Assumption~\ref{ass:admissible-potential}; see
Lemma~\ref{lem:admissible-dual-boundary}.  This yields a uniform logarithmic
Sobolev constant for every finite-time law and hence for $\pi_N$.

\subsection{The boundary space-time LSI}

We next state the boundary space-time LSI, which requires a homogeneous
reference chain and a weakly anharmonic perturbation.

\begin{assump}[Weak anharmonicity with controlled Hessian variation]\label{ass:weak}
The reference system is the homogeneous pinned harmonic chain, and the two
temperatures (which need not be equal) satisfy \eqref{eq:T-bounds}.  We impose
the following two conditions on the perturbation $W_N$:
\begin{enumerate}[label=(\roman*)]
\item its Hessian is small on the relaxation scale of the chain:
   \begin{equation}\label{eq:weak-bound}
      N^3\delta_N=N^3\norm{\nabla^2W_N}_\infty\le\theta_0,
   \end{equation}
   where $\theta_0>0$ is sufficiently small and depends only on
   $\nu,\kappa,\gamma,\underline T,\overline T$;
\item its Hessian has finite uniform variation, that is,
   \[
      \Lambda_N
      :=\sup_{q\in\R^N}\sup_{|h|=1}
      \norm{D(\nabla^2W_N)(q)[h]}_{\mathrm{HS}}<\infty.
   \]
\end{enumerate}
\end{assump}

The first condition makes the anharmonic force a small perturbation over times
of order $N^3$, while the second controls how its linearization changes with
the configuration.  Moreover, for $\theta_0$ sufficiently small,
Assumption~\ref{ass:weak} implies
Assumption~\ref{ass:admissible-potential}.  Indeed, every entry of
$\nabla^2W_N$ is bounded by $\delta_N$, so the bonds of the reference chain remain
nondegenerate, while Lemma~\ref{lem:gramian-bound} and \eqref{eq:weak-bound} give
$\Theta_N\le CN^3\delta_N\le C\theta_0$.

\begin{theorem}[Quantitative boundary space-time LSI]\label{thm:main-stlsi}
Suppose Assumption~\ref{ass:weak} holds.  Then the following statements hold.
\begin{enumerate}[label=(\roman*)]
\item There exist constants
$\tau_c,C>0$, depending only on
$\nu,\kappa,\gamma,\underline T,\overline T$, such that, with $T_c=\tau_cN^3$,
for every $z,z'\in\R^{2N}$,
\begin{equation}\label{eq:kernel-entropy}
   \mathsf{KL}\bigl(\mc P_{T_c}(z',\cdot)\mid\mc P_{T_c}(z,\cdot)\bigr)
      \le C\bigl[1+(N^3\Lambda_N)^2\bigr]
      \norm{z-z'}_E^2 .
\end{equation}
\item There exists a constant
$C>0$, depending only on
$\nu,\kappa,\gamma,\underline T,\overline T$, and an observation time $T_\ast$ satisfying
\begin{equation}\label{eq:TN-bound}
   T_\ast
      \le CN^3
      \bigl[1+\log\bigl(2+(N^3\Lambda_N)^2\bigr)\bigr], 
\end{equation}
such that for every density $f$ with respect to $\pi_N$ and 
$\Ent_{\pi_N}(f)<\infty$, we have 
\begin{equation}
   \Ent_{\pi_N}(f)
      \le2\int_0^{T_\ast}
         \cI_\partial(\mc P_t^*f)\,\dd t. \label{eq:quant-stlsi}
\end{equation}
\item If, in addition,
$\sup_N N^3\Lambda_N<\infty$, then there exist constants $\tau_*,c_*>0$,
depending only on $\nu,\kappa,\gamma,\underline T,\overline T$ and
$\sup_N N^3\Lambda_N$, such that, for every $N$ and every density $f$ with
respect to $\pi_N$ with finite entropy,
\begin{align}
   \Ent_{\pi_N}(f)
      &\le2\int_0^{\tau_*N^3}
          \cI_\partial(\mc P_t^*f)\,\dd t, \label{eq:main-stlsi}\\
   \Ent_{\pi_N}(\mc P_t^*f)
      &\le2e^{-c_*t/N^3}\Ent_{\pi_N}(f),\qquad t\ge0. \label{eq:main-entropy-decay}
\end{align}
\end{enumerate}
\end{theorem}

A concrete class satisfying the additional hypothesis
$\sup_NN^3\Lambda_N<\infty$ in part~(iii) is given by
\[
   W_N(q)=N^{-3}\left(
      \sum_{i=1}^NU(q_i)+\sum_{i=1}^{N-1}R(q_{i+1}-q_i)\right)
\]
with $U'''$ and $R'''$ uniformly bounded.  Indeed, each directional third
derivative involves only a uniformly bounded number of neighboring
coordinates, so the finite-range structure gives $\Lambda_N\le CN^{-3}$ and
hence $\sup_NN^3\Lambda_N<\infty$.

The proof is given in Section~\ref{sec:stlsi}.  We first isolate finitely many
Gaussian components of the boundary noise that, over a time of order $N^3$,
can move the chain in every phase-space direction.  After conditioning on the
rest of the noise, the state at the end of this interval depends smoothly on
these Gaussian variables.  Weak anharmonicity keeps this dependence close to
that of the harmonic chain, while the regularity assumption controls how it
changes with the initial state.  A change of variables then compares the
transition laws from two different initial states and proves
\eqref{eq:kernel-entropy}.  Together with contraction of the dynamics and the
transport inequality for the invariant measure, this comparison shows that
entropy decreases by a uniform proportion after a suitable time interval.  The
entropy-production identity then gives the boundary space-time LSI.

\section{Estimates for harmonic chains}\label{sec:inputs}

Before proving the main results, we collect estimates that quantify how
boundary damping and noise propagate through the harmonic chain.  Two classes
of reference systems appear.  We first define the drift and Gramian for the
general connected, uniformly elliptic tridiagonal reference allowed in
Assumption~\ref{ass:admissible-potential}.  We then specialize to the
homogeneous pinned chain, for which the Gramian, relaxation time, and
controllability time all have the sharp $N^3$ scale.

We use standard finite-dimensional control terminology; see, for instance,
\cite{SontagControl}.  For a general reference in the first class, the
reference equation is linear, with constant drift matrix
\begin{equation}\label{eq:A0}
   A_N^0=\begin{pmatrix}0&I_N\\-K_N^0&-\Gamma_\partial\end{pmatrix}.
\end{equation}
All eigenvalues of $A_N^0$ have negative real part (that is, $A_N^0$ is
Hurwitz).  Indeed, suppose
$A_N^0(q,p)=\lambda(q,p)$ for a nonzero complex eigenvector.  Then
$p=\lambda q$ and
\[
   (\lambda^2I_N+\lambda\Gamma_\partial+K_N^0)q=0 .
\]
Multiplying on the left by $q^*$ gives
\[
   \lambda^2|q|^2+\lambda q^*\Gamma_\partial q
      +q^*K_N^0q=0 .
\]
If $\Im\lambda=0$, the real part rules out $\lambda\ge0$.  If
$\Im\lambda\ne0$, the imaginary part gives
$2(\Re\lambda)|q|^2+q^*\Gamma_\partial q=0$, hence
$\Re\lambda\le0$.  In the equality case, $q^*\Gamma_\partial q=0$,
so $q_1=q_N=0$; the eigenvalue equation reduces to
$K_N^0q=-(\lambda^2)q$, and the boundary recurrence for the tridiagonal
reference chain, with nonzero nearest-neighbor bonds, forces $q=0$, hence
also $p=0$.  Thus no nonzero eigenvector has $\Re\lambda=0$, and $A_N^0$ is
Hurwitz.

We define the Gramian
\begin{equation}
    G_N=\int_0^\infty e^{tA_N^0}e^{t(A_N^0)^{\top}}\dd t.
   \label{eq:GN}
\end{equation}
The matrix $G_N$ is symmetric and positive definite, and satisfies
\begin{equation}\label{eq:GN-Lyapunov}
   A_N^0G_N+G_N(A_N^0)^{\top}=-I_{2N}.
\end{equation}  More concretely, for
every $v\in\R^{2N}$,
\[
   v^\top G_Nv=\int_0^\infty
      \bigl|e^{t(A_N^0)^\top}v\bigr|^2\dd t.
\]
Thus $G_N$ records the total squared size of a solution of the adjoint reference equation
over all times.  Its operator norm is the largest such accumulated
size among unit initial vectors and quantifies the time scale over which
perturbations can build up.

\medskip 

We now specialize to the homogeneous
pinned chain as the reference harmonic chain.  We first derive the deterministic $N^3$ bounds and then the
controlled-path estimate used in Section~\ref{sec:stlsi}.  The first input is a direct
consequence of the Lyapunov construction in
\cite{Menegaki}*{Proposition~1.1 and Lemmas~6.5--6.10}, after translating to our notation.

\begin{lemma}[Gramian estimate]\label{lem:gramian-bound}
For the homogeneous pinned reference chain,
\[
   \norm{G_N}\le CN^3 .
\]
The constant $C$ is independent of $N$ and may depend on
$\nu,\kappa,\gamma$.
\end{lemma}

\begin{theorem}[Relaxation of the homogeneous harmonic chain]\label{thm:harmonic}
For the homogeneous pinned reference chain, there are constants
$M_{\mathrm{h}},\alpha_0>0$, independent of $N$ and depending only on
$\nu,\kappa,\gamma$, such that every solution of the reference equation satisfies
\begin{equation}\label{eq:harmonic-decay}
  \norm{z_t^0}_E\le M_{\mathrm{h}}e^{-\alpha_0t/N^3}\norm{z_0}_E,
  \qquad t\ge0.
\end{equation}
\end{theorem}

\begin{proof}
\emph{Energy balance and contraction.}  Differentiating the energy along a
solution gives
\begin{equation}\label{eq:harmonic-energy-balance}
 \begin{aligned}
   \frac{\dd}{\dd t}\norm{z_t^0}_E^2
      &=2\langle K_N^0q_t^0,p_t^0\rangle
        +2\langle p_t^0,-K_N^0q_t^0-\Gamma_\partial p_t^0\rangle \\
      &=-2p_t^{0,\top}\Gamma_\partial p_t^0
       =-2\gamma\bigl(|p_1^0(t)|^2+|p_N^0(t)|^2\bigr)\le0 .
\end{aligned}
\end{equation}
Writing $U_N^0(t)=e^{tA_N^0}$, the trajectory with initial state $z_0$ is
$z_t^0=U_N^0(t)z_0$.  Since \eqref{eq:harmonic-energy-balance} shows that
$t\mapsto\norm{z_t^0}_E$ is nonincreasing,
\[
   \norm{U_N^0(t)z_0}_E\le\norm{z_0}_E
   \qquad\text{for every }z_0\in\R^{2N}.
\]
Taking the supremum over nonzero $z_0$ gives
\begin{equation}\label{eq:harmonic-contraction}
   \norm{U_N^0(t)}_{E\to E}\le1,\qquad t\ge0 .
\end{equation}

\emph{Relaxation over one window.}  To obtain strict decay, we first seek a
uniform bound on the time integral of the forward trajectory
$e^{tA_N^0}z_0$.  Lemma~\ref{lem:gramian-bound}, however, directly controls
the adjoint trajectory: by \eqref{eq:GN},
\[
   \int_0^\infty\bigl|e^{t(A_N^0)^\top}v\bigr|^2\dd t
      =v^\top G_Nv\le CN^3|v|^2.
\]
We therefore need to relate solutions of the forward and adjoint equations.
For the harmonic chain, time reversal provides this relation.

To make this relation explicit, let
$S=\diag\bigl((K_N^0)^{1/2},I_N\bigr)$, which converts the energy norm into the
Euclidean norm, and let $\mathfrak R=\diag(I_N,-I_N)$ be the momentum flip.
A direct computation gives
\[
   SA_N^0S^{-1}
      =\begin{pmatrix}0&(K_N^0)^{1/2}\\-(K_N^0)^{1/2}&-\Gamma_\partial\end{pmatrix}
\]
and therefore
\[
   (SA_N^0S^{-1})^{\top}=\mathfrak R\,(SA_N^0S^{-1})\,\mathfrak R .
\]
Thus, after passing to energy coordinates, the adjoint dynamics is the forward
dynamics with the momenta reversed.  Transporting the adjoint Gramian $G_N$
through these changes of coordinates produces a Lyapunov matrix for the
forward equation.  Namely, define
\[
   \widetilde G_N=(S\mathfrak R S)G_N(S\mathfrak R S).
\]
For the homogeneous chain,
\[
   \nu I_N\le K_N^0\le(\nu+4\kappa)I_N,
\]
so $S$ and $S^{-1}$ are bounded uniformly in $N$.  The block-diagonal matrices $S$ and
$\mathfrak R$ commute, so a direct substitution using the transpose relation
and \eqref{eq:GN-Lyapunov} gives
\[
   (A_N^0)^{\top}\widetilde G_N+\widetilde G_NA_N^0
      =-S^4\le-cI_{2N}.
\]
Moreover, Lemma~\ref{lem:gramian-bound} and the uniform bound on $S$ yield
\[
   0\le \widetilde G_N
      \le\norm{S}^4\norm{G_N}I_{2N}\le CN^3I_{2N}.
\]
Thus,
\[
   \frac{\dd}{\dd t}(z_t^0)^{\top}\widetilde G_Nz_t^0
      =(z_t^0)^{\top}
        \bigl((A_N^0)^{\top}\widetilde G_N+\widetilde G_NA_N^0\bigr)z_t^0
      \le-c|z_t^0|^2 .
\]
Integrating on $[0,T]$ and using $\widetilde G_N\ge0$ gives
\[
   c\int_0^T |z_s^0|^2\dd s
      \le z_0^{\top}\widetilde G_Nz_0
         -(z_T^0)^{\top}\widetilde G_Nz_T^0
      \le C N^3|z_0|^2 .
\]
By \eqref{eq:K-elliptic}, the Euclidean norm and the energy norm are uniformly
equivalent, so
\[
   \int_0^T\norm{z_s^0}_E^2\dd s\le CN^3\norm{z_0}_E^2,\qquad T\ge0 .
\]
For $0\le s\le T$, the semigroup property and
\eqref{eq:harmonic-contraction} give
$\norm{z_T^0}_E\le\norm{z_s^0}_E$.  Therefore
\[
   T\norm{z_T^0}_E^2\le\int_0^T\norm{z_s^0}_E^2\dd s
      \le CN^3\norm{z_0}_E^2 .
\]
Choose $\tau_0>0$ so large that $\theta:=C/\tau_0\le\tfrac12$.  With
$T=\tau_0N^3$,
\begin{equation}\label{eq:one-window}
   \norm{U_N^0(\tau_0N^3)}_{E\to E}^2\le\theta\le\tfrac12 .
\end{equation}
For general $t\ge0$, write $t=k\tau_0N^3+r$ with $k\ge0$ an integer and
$0\le r<\tau_0N^3$.  Combining \eqref{eq:one-window} with the contraction
\eqref{eq:harmonic-contraction} gives
\[
   \norm{U_N^0(t)}_{E\to E}
      \le \theta^{k/2}
      \le \theta^{-1/2}\exp\!\left(-\frac{\log(1/\theta)}{2\tau_0}\frac{t}{N^3}\right).
\]
This proves \eqref{eq:harmonic-decay}, with
$M_{\mathrm{h}}=\theta^{-1/2}$ and
$\alpha_0=(2\tau_0)^{-1}\log(1/\theta)$.
\end{proof}

We next complement the relaxation estimate with a controllability estimate on
the same $N^3$ time scale.  To this end, we allow an external control to act
through the two boundary momenta.
For $T>0$, an initial state $z_0\in\R^{2N}$, and a control
$u\in L^2([0,T];\R^2)$, the boundary-controlled harmonic reference is
\begin{equation}\label{eq:harmonic-controlled-system}
   \dot z(t)=A_N^0z(t)+B_{\partial,N}u(t),\qquad z(0)=z_0.
\end{equation}
Here $B_{\partial,N}$ is the boundary input matrix from \eqref{eq:sigma}; for
$u=(u_{\mathrm{L}},u_{\mathrm{R}})$,
\[
   B_{\partial,N}u
      =\binom{0}{\sqrt{2\gamma T_1}\,u_{\mathrm{L}}e_1
          +\sqrt{2\gamma T_N}\,u_{\mathrm{R}}e_N}.
\]
We prove the following estimate for this boundary-controlled system.

\begin{prop}[Boundary-controlled harmonic reference chain]\label{prop:harmonic-control}
Under the hypotheses of Theorem~\ref{thm:harmonic}, there are
constants $\tau_c,C_0>0$, depending only on
$\nu,\kappa,\gamma,\underline T,\overline T$, such that the following holds.  Let
$T_c=\tau_cN^3$.  There is a bounded linear operator
\[
   \cR_N^0:\R^{2N}\to L^2([0,T_c];\R^2)
\]
such that, for every target $y\in\R^{2N}$, the control $u=\cR_N^0y$ drives the
solution of \eqref{eq:harmonic-controlled-system} with $z_0=0$ to $y$ at time
$T_c$ and satisfies
\begin{equation}\label{eq:harmonic-control}
  z(T_c)=y,\qquad
  \int_0^{T_c}|u|^2\dd t\le C_0\norm{y}_E^2,\qquad
  \int_0^{T_c}\norm{z(t)}_E^2\dd t\le C_0N^3\norm{y}_E^2 .
\end{equation}
\end{prop}

\begin{proof}
Fix $\tau_c$ large enough that the decay estimate for the harmonic chain from
\eqref{eq:harmonic-decay} satisfies
$M_{\mathrm{h}}^2e^{-2\alpha_0\tau_c}<1/2$.

\emph{A preliminary path to the target.}  We first reverse a freely relaxing
solution of the reference equation.  The reversed damping can be supplied by a boundary
control.  Let $\mathfrak R$ be the momentum flip from the proof of
Theorem~\ref{thm:harmonic}.  For a target $y$, define
\[
   \tilde z_y(t)=\mathfrak R U_N^0(T_c-t)\mathfrak R y,
   \qquad 0\le t\le T_c .
\]
Writing $\tilde z_y=(\tilde q_y,\tilde p_y)$, this path satisfies
$\tilde z_y(T_c)=y$ and
\[
   \dot{\tilde q}_y=\tilde p_y,\qquad
   \dot{\tilde p}_y=-K_N^0\tilde q_y+\Gamma_\partial\tilde p_y .
\]
Compared with the uncontrolled equation, the damping has changed from
$-\Gamma_\partial\tilde p_y$ to $+\Gamma_\partial\tilde p_y$.  The difference
is produced by the boundary control
\[
   \tilde u_{y,{\mathrm{L}}}(t)=\sqrt{\frac{2\gamma}{T_1}}\,\tilde p_{y,1}(t),
   \qquad
   \tilde u_{y,{\mathrm{R}}}(t)=\sqrt{\frac{2\gamma}{T_N}}\,\tilde p_{y,N}(t),
\]
because $B_{\partial,N}\tilde u_y=(0,2\Gamma_\partial\tilde p_y)$.  Thus
$\tilde z_y$ reaches $y$ at time $T_c$, although it does not generally start
from the origin.

Since the momentum flip preserves both the energy norm and
$p^{\top}\Gamma_\partial p$, the energy balance
\eqref{eq:harmonic-energy-balance}, applied to
$s\mapsto U_N^0(s)\mathfrak R y$, gives
\[
   \int_0^{T_c}|\tilde u_y(t)|^2\dd t
      \le\frac{2}{\underline T}\int_0^{T_c}\tilde p_y^{\top}\Gamma_\partial\tilde p_y\dd t
      \le\frac{1}{\underline T}\norm{y}_E^2 ,
\]
and \eqref{eq:harmonic-decay} gives
\[
   \int_0^{T_c}\norm{\tilde z_y(t)}_E^2\dd t
      =\int_0^{T_c}\norm{U_N^0(s)\mathfrak R y}_E^2\dd s
      \le C N^3\norm{y}_E^2 .
\]

\emph{Correction of the initial point.}  Introduce the momentum-reversed
solution operator
\[
   \widehat U_N^0(T_c):=\mathfrak R U_N^0(T_c)\mathfrak R ,
\]
so that $\tilde z_y(0)=\widehat U_N^0(T_c)y$.  Moreover,
$\norm{\widehat U_N^0(T_c)}_{E\to E}\le M_{\mathrm{h}}e^{-\alpha_0\tau_c}$.  If the
same control $\tilde u_y$ is instead applied from the origin, the terminal
point becomes
\[
   \bigl[I_{2N}-U_N^0(T_c)\widehat U_N^0(T_c)\bigr]y.
\]
By \eqref{eq:harmonic-decay} and our choice of $\tau_c$,
\[
   \norm{U_N^0(T_c)\widehat U_N^0(T_c)}_{E\to E}
      \le M_{\mathrm{h}}^2e^{-2\alpha_0\tau_c}
   <\frac12 .
\]
Hence $I_{2N}-U_N^0(T_c)\widehat U_N^0(T_c)$ is invertible with uniformly bounded
inverse.  To reach a prescribed target $y$, set
\[
   x=\bigl[I_{2N}-U_N^0(T_c)\widehat U_N^0(T_c)\bigr]^{-1}y
\]
and use the preliminary control associated with $x$:
\[
   \cR_N^0y:=\tilde u_x.
\]
The corresponding trajectory from the origin is
\begin{equation}\label{eq:corrected-control-path}
   z(t)=\tilde z_x(t)-U_N^0(t)\widehat U_N^0(T_c)x.
\end{equation}
At $t=0$ the two terms cancel, while at $t=T_c$ their difference is
$\bigl[I_{2N}-U_N^0(T_c)\widehat U_N^0(T_c)\bigr]x=y$.  Hence the control steers the
system from the origin to the prescribed target.

\emph{Control and path bounds.}  Since
$\norm{x}_E\le C\norm{y}_E$, the preliminary control estimate gives
\[
   \int_0^{T_c}|\cR_N^0y(t)|^2\dd t
      \le C\norm{y}_E^2 .
\]
For the second term in \eqref{eq:corrected-control-path},
\eqref{eq:harmonic-decay} gives
\[
   \int_0^{T_c}
      \norm{U_N^0(t)\widehat U_N^0(T_c)x}_E^2\dd t
      \le CN^3\norm{x}_E^2.
\]
Combining this with the path estimate for $\tilde z_x$ yields
\[
   \int_0^{T_c}\norm{z(t)}_E^2\dd t\le C N^3\norm{y}_E^2 .
\]
Increasing the constant if necessary gives \eqref{eq:harmonic-control}.
\end{proof}

\section{Proof of the dimension-free full-gradient LSI}\label{sec:lsi}

The proof transfers the Gaussian LSI from finite-dimensional
approximations of the boundary noise to every finite-time transition law and
then to the NESS\@.  The key input is a horizon- and dimension-uniform boundary
estimate for the adjoint propagator, based only on the first
variational equation and the bounded Hessian of the perturbation.

We first record the energy identity for the adjoint reference equation.  Consider the
dual energy $\frac12\norm{w}_{E,*}^2$ from
\eqref{eq:energy-dual-norms} along this equation
\begin{equation}\label{eq:dual-harmonic-flow}
   \dot w=(A_N^0)^{\top}w,\qquad
   \dot w_q=-K_N^0w_p,\qquad
   \dot w_p=w_q-\Gamma_\partial w_p .
\end{equation}
The bulk terms cancel in the energy balance, leaving only the boundary
dissipation.

\begin{prop}[Dual boundary-energy identity]\label{prop:dual-boundary-energy}
For every solution of \eqref{eq:dual-harmonic-flow}, one has
\begin{equation}\label{eq:Vdot-harm}
   \frac12\frac{\dd}{\dd s}\norm{w(s)}_{E,*}^2
      =-\ip{w_p}{\Gamma_\partial w_p}.
\end{equation}
\end{prop}

\begin{proof}
Differentiating the dual energy along \eqref{eq:dual-harmonic-flow} gives
\[
   \frac12\frac{\dd}{\dd s}\norm{w}_{E,*}^2
   =\ip{w_p}{\dot w_p}+\ip{(K_N^0)^{-1}w_q}{\dot w_q}.
\]
Substituting $\dot w_q=-K_N^0w_p$ and
$\dot w_p=w_q-\Gamma_\partial w_p$, the Hamiltonian cross terms cancel:
\[
   \ip{w_p}{w_q-\Gamma_\partial w_p}
      -\ip{(K_N^0)^{-1}w_q}{K_N^0w_p}
   =-\ip{w_p}{\Gamma_\partial w_p},
\]
which proves \eqref{eq:Vdot-harm}.
\end{proof}

We now extend this identity to the variational equations generated by the
anharmonic chain.  The first step is a boundary estimate for the corresponding
adjoint propagator, uniform in both the time horizon and the chain length.

\begin{lemma}[Boundary estimate for the adjoint propagator]\label{lem:admissible-dual-boundary}
Assume Assumption~\ref{ass:admissible-potential}, with $\theta_0$ chosen
sufficiently small.  Let $K_t=K_t^{\top}$ be measurable with
$\norm{K_t}_{\mathrm{op}}\le\delta_N$, and let $U_K(t,s)$ be the evolution operator of
\[
   \dot z_t=\left(A_N^0+
      \begin{pmatrix}0&0\\-K_t&0\end{pmatrix}\right)z_t,
   \qquad t\ge s.
\]
Then, for every $0\le r<t$ and $v\in\R^{2N}$,
\begin{equation}\label{eq:admissible-dual-boundary}
   \int_r^t
      \left|B_{\partial,N}^{\top}U_K(t,s)^{\top}v\right|^2\dd s
      \le C_{\mathrm{B}}\norm{v}_{E,*}^2,
\end{equation}
where $C_{\mathrm{B}}$ depends only on $m_0,M_0,\gamma,\underline T,\overline T$ and is
independent of $N,r,t$, and the coefficient path $K$.
\end{lemma}

\begin{proof}
Put $T=t-r$ and
\[
   w(\tau)=U_K(t,t-\tau)^{\top}v,
   \qquad 0\le\tau\le T.
\]
Then $w(0)=v$ and
\[
   \dot w(\tau)=
   \left(A_N^0+
      \begin{pmatrix}0&0\\-K_{t-\tau}&0\end{pmatrix}\right)^{\top}w(\tau).
\]
We first control its time-integrated Euclidean norm.  Since $G_N$ solves
\eqref{eq:GN-Lyapunov},
\begin{align*}
   \frac{\dd}{\dd\tau}\bigl(w^{\top}G_Nw\bigr)
   &=-|w|^2
     +w^{\top}\left[
        \begin{pmatrix}0&0\\-K_{t-\tau}&0\end{pmatrix}G_N
        +G_N
        \begin{pmatrix}0&-K_{t-\tau}\\0&0\end{pmatrix}
       \right]w\\
   &\le -(1-2\delta_N\norm{G_N})|w|^2
    \le-\tfrac12|w|^2,
\end{align*}
where the last inequality follows from
$\delta_N\norm{G_N}\le\Theta_N\le\theta_0$ after choosing $\theta_0\le1/4$.  Integration gives
\begin{equation}\label{eq:adjoint-integrated-euclidean}
   \int_0^T|w(\tau)|^2\dd\tau
      \le2v^{\top}G_Nv
      \le2\norm{G_N}\,|v|^2.
\end{equation}

Write $w=(w_q,w_p)$.  Repeating the dual-energy calculation of
Proposition~\ref{prop:dual-boundary-energy}, now with the Hessian perturbation,
gives
\[
   \frac12\frac{\dd}{\dd\tau}\norm{w(\tau)}_{E,*}^2
   =-\langle w_p,\Gamma_\partial w_p\rangle
     -w_q^{\top}(K_N^0)^{-1}K_{t-\tau}w_p .
\]
Integrating on $[0,T]$, discarding the nonnegative terminal energy, and using
\eqref{eq:adjoint-integrated-euclidean},
\begin{align*}
   \int_0^T\langle w_p,\Gamma_\partial w_p\rangle\dd\tau
   &\le\tfrac12\norm{v}_{E,*}^2
      +\frac{\delta_N}{m_0}\int_0^T|w_q|\,|w_p|\dd\tau\\
   &\le\tfrac12\norm{v}_{E,*}^2
      +C\delta_N\norm{G_N}\,|v|^2\\
   &\le C\norm{v}_{E,*}^2.
\end{align*}
In the last step we used the norm equivalence \eqref{eq:norm-equivalence}, in
the form $|v|^2\le\max\{M_0,1\}\norm{v}_{E,*}^2$, and
$\delta_N\norm{G_N}\le\theta_0$.  Finally,
\[
   |B_{\partial,N}^{\top}w|^2
   =2\gamma\sum_{i\in\partial}T_i|w_{p,i}|^2
   \le2\overline T\langle w_p,\Gamma_\partial w_p\rangle.
\]
Changing variables $s=t-\tau$ proves
\eqref{eq:admissible-dual-boundary}.
\end{proof}

We next transfer the Gaussian logarithmic Sobolev inequality to each
finite-time transition law.  Following the standard finite-dimensional
approximation of Wiener space underlying Gross's inequality
\cite{Gross1975}, we approximate the Brownian paths by polygonal paths
depending on finitely many Gaussian coordinates.  This allows us to
differentiate with respect to the noise in a finite-dimensional setting.

\begin{lemma}[Uniform LSI for finite-time transition laws]\label{lem:finite-horizon-lsi}
Fix $T>0$ and $z\in\R^{2N}$.
Under Assumption~\ref{ass:admissible-potential}, every
$u\in C_c^\infty(\R^{2N})$ satisfies
\begin{equation}\label{eq:finite-horizon-lsi}
   \Ent_{\mc P_T(z,\cdot)}(u^2)
      \le2C_{\mathrm{B}}\int\norm{\nabla u(z')}_{E,*}^2\,\mc P_T(z,\dd z'),
\end{equation}
with the same constant $C_{\mathrm{B}}$ as in
Lemma~\ref{lem:admissible-dual-boundary}.  In particular, the constant is
independent of $T,z$, and $N$.
\end{lemma}

\begin{proof}
Since $\nabla^2V_N=K_N^0+\nabla^2W_N$ is bounded, the drift $b$ in
\eqref{eq:SDE-vector} is globally Lipschitz, and the solution satisfies the
pathwise integral equation
\begin{equation}\label{eq:additive-pathwise-equation}
   Z_t=z+\int_0^t b(Z_s)\dd s+B_{\partial,N}W_t.
\end{equation}

Let $(e_j)_{j\ge1}$ be the dyadic Haar orthonormal basis of
$L^2([0,T];\R^2)$, ordered so that the primitives
\[
   h_j(t)=\int_0^t e_j(s)\dd s
\]
give the L\'evy--Ciesielski expansion.  If $(\xi_j)_{j\ge1}$ are independent
standard Gaussian variables, then the polygonal paths
\[
   W_t^{(m)}=\sum_{j=1}^m\xi_jh_j(t)
\]
converge to $W$ uniformly on $[0,T]$ almost surely.  Let $Z^{(m)}$ solve
\eqref{eq:additive-pathwise-equation} with $W$ replaced by $W^{(m)}$.  Global
Lipschitz continuity of $b$, followed by Gronwall's inequality, gives
\begin{equation}\label{eq:polygonal-solution-convergence}
   \sup_{0\le t\le T}|Z_t^{(m)}-Z_t|
      \le \norm{B_{\partial,N}}e^{T\norm{Db}_\infty}
          \sup_{0\le t\le T}|W_t^{(m)}-W_t|
      \longrightarrow0
\end{equation}
almost surely.

For fixed $m$, the state $Z_T^{(m)}$ at time $T$ is a $C^1$ function of
$\xi=(\xi_1,\ldots,\xi_m)$, whose law is
$\gamma_m:=\mathcal{N}(0,I_m)$.  Set
$K_s^{(m)}=\nabla^2W_N(q_s^{(m)})$, so that the linearized propagator along
$Z^{(m)}$ is $U_{K^{(m)}}(t,s)$ in the notation of
Lemma~\ref{lem:admissible-dual-boundary}.  Differentiating the integral
equation with respect to $\xi_j$ yields
\[
   \partial_{\xi_j}Z_T^{(m)}
      =\int_0^T U_{K^{(m)}}(T,s)B_{\partial,N}e_j(s)\dd s.
\]
Apply Gross's Gaussian logarithmic Sobolev inequality \cite{Gross1975} on
$\R^m$ to the function $\xi\mapsto u(Z_T^{(m)}(\xi))$.  The preceding identity
and Bessel's inequality in $L^2([0,T];\R^2)$ give
\begin{align*}
   \Ent_{\gamma_m}\bigl(u(Z_T^{(m)}(\xi))^2\bigr)
   &\le2\E_{\gamma_m}\sum_{j=1}^m
      \bigl|\partial_{\xi_j}u(Z_T^{(m)})\bigr|^2\\
   &\le2\E_{\gamma_m}\int_0^T
      \left|B_{\partial,N}^{\top}U_{K^{(m)}}(T,s)^{\top}
         \nabla u(Z_T^{(m)})\right|^2\dd s .
\end{align*}
For each realization of the Gaussian coordinates, the transpose
$U_{K^{(m)}}(T,s)^{\top}\nabla u(Z_T^{(m)})$ solves the adjoint equation
appearing in Lemma~\ref{lem:admissible-dual-boundary}.  Since $K_s^{(m)}$ is
symmetric and $\norm{K_s^{(m)}}_{\mathrm{op}}\le\delta_N$, the lemma applies
pathwise and yields
\begin{equation}\label{eq:finite-gaussian-LSI-bound}
   \Ent_{\gamma_m}\bigl(u(Z_T^{(m)}(\xi))^2\bigr)
      \le2C_{\mathrm{B}}\E_{\gamma_m}\norm{\nabla u(Z_T^{(m)}(\xi))}_{E,*}^2.
\end{equation}

It remains to let $m\to\infty$.  Since $u$ and $\nabla u$ are bounded and
continuous, \eqref{eq:polygonal-solution-convergence} and dominated
convergence give
\[
   u(Z_T^{(m)})\longrightarrow u(Z_T)\quad\text{in every }L^p,
   \qquad
   \E\norm{\nabla u(Z_T^{(m)})}_{E,*}^2
      \longrightarrow\E\norm{\nabla u(Z_T)}_{E,*}^2.
\]
The entropies converge as well: this is immediate if $\E u(Z_T)^2=0$; otherwise
$\E u(Z_T^{(m)})^2$ stays bounded away from zero, and dominated convergence
applies to
\[
   u(Z_T^{(m)})^2
   \log\frac{u(Z_T^{(m)})^2}{\E u(Z_T^{(m)})^2}.
\]
Passing to the limit in
\eqref{eq:finite-gaussian-LSI-bound} proves
\eqref{eq:finite-horizon-lsi}.
\end{proof}

\begin{proof}[Proof of Theorem~\ref{thm:main-lsi}]
Fix $z\in\R^{2N}$.  Proposition~\ref{prop:fixedN-mixing} gives
\[
   \mc P_T(z,\cdot)\longrightarrow\pi_N
   \qquad\text{in total variation as }T\to\infty.
\]
For $u\in C_c^\infty(\R^{2N})$, the functions $u^2$, $u^2\log u^2$, and
$\norm{\nabla u}_{E,*}^2$ are bounded.  Hence
Lemma~\ref{lem:finite-horizon-lsi} passes to the limit and gives
\begin{equation}\label{eq:NESS-function-LSI-energy}
   \Ent_{\pi_N}(u^2)
      \le2C_{\mathrm{B}}\int\norm{\nabla u}_{E,*}^2\dd\pi_N
      \le C\int|\nabla u|^2\dd\pi_N,
\end{equation}
where the last constant is uniform by the norm equivalence
\eqref{eq:norm-equivalence}, in the form
$\norm{v}_{E,*}^2\le\max\{1/m_0,1\}|v|^2$.

We extend \eqref{eq:NESS-function-LSI-energy} to the domain of the closed full
Sobolev form introduced after \eqref{eq:full-fisher}.  By definition, for every
such $u$ there are $u_n\in C_c^\infty(\R^{2N})$ with $u_n\to u$ in
$L^2(\pi_N)$ and
$\int|\nabla(u_n-u)|^2\dd\pi_N\to0$.  Hence $u_n^2\to u^2$ in
$L^1(\pi_N)$.  Relative entropy is lower semicontinuous for $L^1$ convergence
of nonnegative functions, while the Dirichlet integrals converge.  Passing to
the limit proves \eqref{eq:NESS-function-LSI-energy} throughout the closed form
domain.

For a density $f$, put $u=\sqrt f$.  If
$\cI_{\mathrm{full}}(f)=4\int|\nabla u|^2\dd\pi_N=+\infty$, there is nothing to
prove.  Otherwise $u$ belongs to the closed full Sobolev form domain and $\int u^2\dd\pi_N=1$, so
\eqref{eq:NESS-function-LSI-energy} gives
\[
   \Ent_{\pi_N}(f)
      \le C_{\mathrm{LSI}}\,\cI_{\mathrm{full}}(f),
\]
with $C_{\mathrm{LSI}}$ independent of $N$.  This proves
\eqref{eq:main-lsi}.
\end{proof}

\section{Proof of the boundary space-time LSI}\label{sec:stlsi}

We prove Theorem~\ref{thm:main-stlsi} by showing that entropy decreases by a
uniform proportion over a suitable time interval.  The argument has three stages.
First, the relaxation estimate for the harmonic chain is shown to persist under weak anharmonicity.
Second, we extract $2N$ Gaussian directions from the boundary noise and use a
conditional change of variables to compare transition kernels in relative
entropy.  Third, this kernel estimate is combined with Wasserstein contraction
and Talagrand inequality to obtain the entropy contraction.

\subsection{Robust deterministic estimates under weak anharmonicity}\label{sec:robust}

Section~\ref{sec:inputs} gave relaxation and controllability on the harmonic
$N^3$ time scale.  Here we show that the relaxation estimate is stable under
the weak Hessian perturbation in Assumption~\ref{ass:weak}.  Let
$K_t=K_t^\top$ be measurable, and consider the linear equation for
$z_t=(x_t,y_t)$:
\begin{equation}\label{eq:linear-family}
   \dot x_t=y_t,\qquad\dot y_t=-(K_N^0+K_t)x_t-\Gamma_\partial y_t,
   \qquad\norm{K_t}_{\mathrm{op}}\le\delta_N,\ \ N^3\delta_N\le\theta_0,
\end{equation}
with evolution operator $U_K(t,s)$.  For $z=(x,y)$, the perturbation of the
reference drift satisfies
\[
   \begin{pmatrix}0&0\\-K_t&0\end{pmatrix}z=(0,-K_tx),
   \qquad
   \norm{(0,-K_tx)}_E\le C\delta_N\norm{z}_E .
\]

Equation~\eqref{eq:linear-family} arises naturally from the nonlinear
dynamics.  Indeed, if $(q_t,p_t)$ and $(q'_t,p'_t)$ are two solutions of
\eqref{eq:SDE} driven by the same Brownian motion, then their difference
$z_t=(q_t-q'_t,p_t-p'_t)$ solves \eqref{eq:linear-family} with
\[
   K_t=\int_0^1\nabla^2W_N(q'_t+s(q_t-q'_t))\dd s .
\]
The reference dynamics decays at rate $O(N^{-3})$, whereas the perturbation of
its drift has size $O(\delta_N)$.  Stability therefore requires $\delta_N$ to
be small compared with $N^{-3}$, which is exactly the condition
$N^3\delta_N\le\theta_0$ in Assumption~\ref{ass:weak}.

\begin{lemma}[Robust exponential stability]\label{lem:robust-stability}
If $N^3\delta_N\le\theta_0$, then
\begin{equation}\label{eq:robust-decay}
   \norm{U_K(t,s)z}_E
      \le M_{\mathrm{h}}e^{-\alpha_0(t-s)/(2N^3)}\norm{z}_E,
      \qquad 0\le s\le t.
\end{equation}
\end{lemma}

\begin{proof}
Put $a_N=\alpha_0/N^3$.  Duhamel's formula and
\eqref{eq:harmonic-decay} give, with $f(t)=\norm{U_K(t,s)z}_E$,
\[
   f(t)\le M_{\mathrm{h}}e^{-a_N(t-s)}\norm{z}_E
      +C\delta_NM_{\mathrm{h}}\int_s^t e^{-a_N(t-r)}f(r)\,\dd r.
\]
Set $g(t)=e^{a_N(t-s)}f(t)$.  Then
\[
   g(t)\le M_{\mathrm{h}}\norm{z}_E
      +C\delta_NM_{\mathrm{h}}\int_s^t g(r)\,\dd r.
\]
Gronwall's inequality therefore yields
\[
   f(t)\le M_{\mathrm{h}}
      \exp\!\left[-\bigl(a_N-C\delta_NM_{\mathrm{h}}\bigr)(t-s)\right]
      \norm{z}_E.
\]
Choose $\theta_0$ so that
$C M_{\mathrm{h}}N^3\delta_N\le\alpha_0/2$.  This proves
\eqref{eq:robust-decay}.
\end{proof}

\subsection{Finite-dimensional Gaussian coordinates in the boundary noise}\label{sec:gaussian-component}

We now turn to the transition-kernel estimate.  The goal is to choose $2N$
orthonormal directions in the boundary-noise path space whose images under the
reference controllability operator form a uniformly well-conditioned basis of
phase space.

Let
\[
   \mathcal B_N^0u
      :=\int_0^{T_c}U_N^0(T_c-s)B_{\partial,N}u_s\,\dd s
\]
be the controllability operator for the reference equation at time $T_c$.
Thus $\mathcal B_N^0u$ is the state at time $T_c$ of the reference equation
started from zero and driven by the boundary control $u$.  Proposition~\ref{prop:harmonic-control}
provides a right inverse
$\cR_N^0:\R^{2N}\to L^2([0,T_c];\R^2)$ with
$\mathcal B_N^0\cR_N^0=I_{2N}$.

The right inverse need not be isometric.  The next lemma orthonormalizes its
range without losing the uniform bounds on the resulting trajectories and
their states at time $T_c$.  The operator $\mathcal S_N$ maps phase-space coordinates
to orthonormal control directions, while $\mathcal M_N=\mathcal B_N^0\mathcal S_N$ maps those
coordinates to the corresponding states at time $T_c$.  Uniform bounds on
$\mathcal M_N$ and $\mathcal M_N^{-1}$ ensure that these states form a well-conditioned basis of
phase space.

\begin{lemma}[Normalized control directions for the reference equation]\label{lem:normalized-controls}
There is a constant $C>0$, depending only on
$\nu,\kappa,\gamma,\underline T,\overline T$, such that, for every $N$, there
is an operator $\mathcal S_N:\R^{2N}\to L^2([0,T_c];\R^2)$ for which, with
$\mathcal M_N:=\mathcal B_N^0\mathcal S_N$,
\begin{equation}\label{eq:SN-MN}
   \mathcal S_N^*\mathcal S_N=I_{2N},
   \qquad
   \norm{\mathcal M_N}_{E\to E}
      +\norm{\mathcal M_N^{-1}}_{E\to E}\le C.
\end{equation}
If $z_a^0(t)$ is the solution of \eqref{eq:harmonic-controlled-system} with
$z_a^0(0)=0$ and control $\mathcal S_Na$, then
\begin{equation}\label{eq:gaussian-direction-path-bound}
   \int_0^{T_c}\norm{z_a^0(t)}_E^2\,\dd t
      \le CN^3\norm{a}_E^2.
\end{equation}
\end{lemma}

\begin{proof}
Lemma~\ref{lem:admissible-dual-boundary}, applied with $K=0$, $r=0$, and
$t=T_c$, gives by duality
\[
\begin{aligned}
   \norm{\mathcal B_N^0u}_E
      &=\sup_{\norm{v}_{E,*}=1}
         \int_0^{T_c}
          \left\langle B_{\partial,N}^{\top}
             U_N^0(T_c-s)^{\top}v,u_s\right\rangle\dd s\\
      &\le C_{\mathrm{B}}^{1/2}\norm{u}_{L^2([0,T_c];\R^2)}.
\end{aligned}
\]
Thus $\mathcal B_N^0$ is uniformly bounded from
$L^2([0,T_c];\R^2)$ to $\R^{2N}$ equipped with the energy norm.  The control bound in
Proposition~\ref{prop:harmonic-control} gives
\[
   (\cR_N^0)^*\cR_N^0\le CI_{2N}.
\]
Conversely,
\[
   \norm{y}_E
      =\norm{\mathcal B_N^0\cR_N^0y}_E
      \le C\norm{\cR_N^0y}_{L^2([0,T_c];\R^2)},
\]
so $cI_{2N}\le(\cR_N^0)^*\cR_N^0$.  Define
\[
   \mathcal S_N=\cR_N^0\bigl((\cR_N^0)^*\cR_N^0\bigr)^{-1/2},
   \qquad
   \mathcal M_N=\mathcal B_N^0\mathcal S_N
      =\bigl((\cR_N^0)^*\cR_N^0\bigr)^{-1/2}.
\]
This proves \eqref{eq:SN-MN}.  The control $\mathcal S_Na$ equals
$\cR_N^0\bigl((\cR_N^0)^*\cR_N^0\bigr)^{-1/2}a$, and
\eqref{eq:gaussian-direction-path-bound} follows from the path estimate in
\eqref{eq:harmonic-control} and the uniform bound on
$\bigl((\cR_N^0)^*\cR_N^0\bigr)^{-1/2}$.
\end{proof}

Choose a basis $(e_j)_{j=1}^{2N}$ of $\R^{2N}$ orthonormal with respect to
$\langle\cdot,\cdot\rangle_E$ and let
$u_j=\mathcal S_Ne_j\in L^2([0,T_c];\R^2)$.  Since
$\mathcal S_N^*\mathcal S_N=I_{2N}$, the $u_j$ are orthonormal in
$L^2([0,T_c];\R^2)$.  On the Wiener space over $[0,T_c]$, define
\begin{equation}\label{eq:xi-def}
   \xi_j=\int_0^{T_c}u_j(s)\cdot\dd W_s,
   \qquad
   \xi=(\xi_1,\ldots,\xi_{2N})\sim\gamma_{2N},
   \qquad \gamma_{2N}:=\mathcal{N}(0,I_{2N}),
\end{equation}
and let $h_j(t)=\int_0^tu_j(s)\,\dd s$.  The orthogonal Gaussian residual
\begin{equation}\label{eq:Wperp}
   W_t^\perp=W_t-\sum_{j=1}^{2N}\xi_jh_j(t)
\end{equation}
is independent of $\xi$.  Indeed, the pair $(W^\perp,\xi)$ is jointly Gaussian
and $\operatorname{Cov}(W_t^\perp,\xi_j)=0$ for every $t,j$.  Thus
\eqref{eq:Wperp} decomposes the full noise as
\[
   W=W^\perp+\sum_{j=1}^{2N}\xi_jh_j,
\]
where the residual path $W^\perp$ and the standard Gaussian vector $\xi$ are
independent.

Fix a realization of the residual path $\omega\in C_0([0,T_c];\R^2)$ and an initial state
$z\in\R^{2N}$.  We regard the state at time $T_c$ as a function of the Gaussian coordinates $\xi\in\R^{2N}$.  Let
$Z_t^{z,\omega,\xi}$ denote the trajectory starting from $z$ and driven by
$\omega+\sum_j\xi_jh_j$, and define
\[
   F_z^\omega(\xi):=Z_{T_c}^{z,\omega,\xi}.
\]
Thus $F_z^\omega$ maps the selected Gaussian coordinates to the state at time
$T_c$.  Because the noise is additive and the drift is globally Lipschitz,
the pathwise equation is well posed for every continuous $\omega$; the
smoothness of the anharmonic perturbation makes $F_z^\omega$ a $C^2$ function
of $\xi$.

We next express this state in the coordinates determined by the reference
controls.  Using the basis $(e_j)_{j=1}^{2N}$ fixed above to identify phase
space with $\R^{2N}$, set
\begin{equation}\label{eq:Phi-def}
   \Phi_z^\omega(\xi)=\mathcal M_N^{-1}F_z^\omega(\xi).
\end{equation}
Because $(e_j)$ is orthonormal for $\langle\cdot,\cdot\rangle_E$, in these
coordinates the energy inner product is the Euclidean one; consequently
$\norm{\cdot}_{E\to E}$ and $\norm{\cdot}_{\mathrm{HS},E}$ coincide with the
Euclidean operator and Hilbert--Schmidt norms, and $\xi\sim\gamma_{2N}$ is a
standard Gaussian vector. 
For the harmonic chain, the derivative of $F_z^\omega$ with respect to $\xi$
is $\mathcal M_N$, so $D\Phi_z^\omega=I_{2N}$.  Weak anharmonicity keeps this
derivative close to the identity.  This allows us to compare the conditional
laws at time $T_c$ by a finite-dimensional change of variables.
The construction is motivated by Malliavin controllability, but after
conditioning on the residual path, the argument involves only an ordinary
finite-dimensional Gaussian map.

\begin{lemma}[Uniform conditional diffeomorphism]\label{lem:conditional-diffeo}
For $\theta_0$ in Assumption~\ref{ass:weak} sufficiently small, for all
$z,\omega,\xi$,
\begin{equation}\label{eq:Dphi-close}
   \norm{D\Phi_z^\omega(\xi)-I_{2N}}_{E\to E}\le\frac1{10}.
\end{equation}
Consequently every $\Phi_z^\omega$ is an orientation-preserving global $C^1$
diffeomorphism.
\end{lemma}

\begin{proof}
Let $q_t$ denote the position component of the trajectory
$Z_t^{z,\omega,\xi}$ defined above.  Define its Jacobian with respect to the
Gaussian coordinates by
\[
   \cJ_t:=D_\xi Z_t^{z,\omega,\xi}:\R^{2N}\to\R^{2N}.
\]
It solves
\[
   \dot{\cJ}_t=A_t\cJ_t+B_{\partial,N}\mathcal S_N(t),
   \qquad \cJ_0=0,
\]
where $A_t$ is of the form \eqref{eq:linear-family} with
$K_t=\nabla^2W_N(q_t)$ and $\norm{K_t}_{\mathrm{op}}\le\delta_N$.  The reference
trajectory corresponding to a direction $a\in\R^{2N}$ is $z_a^0(t)$ from
Lemma~\ref{lem:normalized-controls}.  Variation of constants gives
\[
   \cJ_{T_c}a-z_a^0(T_c)
      =\int_0^{T_c}U_K(T_c,r)
         \binom{0}{-K_r(z_a^0(r))_q}\,\dd r.
\]
Lemma~\ref{lem:robust-stability} bounds the propagator in this integral by
\[
   \norm{U_K(T_c,r)}_{E\to E}
      \le M_{\mathrm{h}}e^{-\alpha_0(T_c-r)/(2N^3)}.
\]
Together with
$\norm{(0,-K_r(z_a^0(r))_q)}_E\le C\delta_N\norm{z_a^0(r)}_E$, this gives
\[
\begin{aligned}
   \norm{\cJ_{T_c}a-z_a^0(T_c)}_E
   &\le C\delta_N\int_0^{T_c}
      e^{-\alpha_0(T_c-r)/(2N^3)}\norm{z_a^0(r)}_E\,\dd r\\
   &\le C\delta_N
      \left(\int_0^{T_c}e^{-\alpha_0(T_c-r)/N^3}\,\dd r\right)^{1/2}
      \left(\int_0^{T_c}\norm{z_a^0(r)}_E^2\,\dd r\right)^{1/2} \\
   &\le CN^3\delta_N\norm{a}_E.
\end{aligned}
\]
Since $z_a^0(T_c)=\mathcal M_Na$, multiplication by the uniformly bounded $\mathcal M_N^{-1}$
gives \eqref{eq:Dphi-close} after decreasing $\theta_0$.

The bound \eqref{eq:Dphi-close} and the contraction mapping theorem, applied
to $\xi\mapsto y+\xi-\Phi_z^\omega(\xi)$, show that $\Phi_z^\omega$ is a global
$C^1$ diffeomorphism.  Moreover,
$I_{2N}+t(D\Phi_z^\omega-I_{2N})$ is invertible for $0\le t\le1$, so
$\det D\Phi_z^\omega>0$.
\end{proof}

To compare the conditional laws from two initial states, we need to control
both the displacement of the maps $\Phi_z^\omega$ and the variation of their
Jacobians.  The latter enters the Gaussian change-of-variables formula through
a log-determinant term and is therefore estimated in Hilbert--Schmidt norm.
Here $\norm{\cdot}_{\mathrm{HS},E}$ denotes the Hilbert--Schmidt norm with respect
to $\langle\cdot,\cdot\rangle_E$; it is uniformly equivalent to the Euclidean
Hilbert--Schmidt norm.

\begin{lemma}[Dependence on the initial state and Gaussian coordinates]\label{lem:Phi-derivative-Lipschitz}
For every fixed residual path $\omega$ and all $z,z'\in\R^{2N}$ and
$\xi,\eta\in\R^{2N}$,
\begin{align}
   \sup_{0\le t\le T_c}
      \norm{Z_t^{z,\omega,\xi}-Z_t^{z',\omega,\eta}}_E
      &\le C\bigl(\norm{z-z'}_E+|\xi-\eta|\bigr),
      \label{eq:path-difference}\\
   |\Phi_z^\omega(\xi)-\Phi_{z'}^\omega(\xi)|
      &\le C\norm{z-z'}_E, \label{eq:Phi-start}\\
   \norm{D\Phi_z^\omega(\xi)-D\Phi_{z'}^\omega(\eta)}_{\mathrm{HS},E}
      &\le CN^3\Lambda_N
         \bigl(\norm{z-z'}_E+|\xi-\eta|\bigr).
      \label{eq:Dphi-HS}
\end{align}
\end{lemma}

\begin{proof}
Set $Z_1(t)=Z_t^{z,\omega,\xi}$ and $Z_2(t)=Z_t^{z',\omega,\eta}$, and denote
their position components by $q_1(t)$ and $q_2(t)$, respectively.  Their
difference solves
\[
   \frac{\dd}{\dd t}\bigl(Z_1(t)-Z_2(t)\bigr)
      =\widetilde A_t\bigl(Z_1(t)-Z_2(t)\bigr)
      +B_{\partial,N}\mathcal S_N(\xi-\eta)(t),
   \qquad Z_1(0)-Z_2(0)=z-z',
\]
where $\widetilde A_t$ has the form \eqref{eq:linear-family}; denote its
Hessian perturbation by $\widetilde K_t$, the mean of $\nabla^2W_N$ along the
line segment joining the two position components.  Duhamel's formula gives
\[
   Z_1(t)-Z_2(t)=U_{\widetilde K}(t,0)(z-z')
      +\int_0^tU_{\widetilde K}(t,s)B_{\partial,N}
         \mathcal S_N(\xi-\eta)(s)\,\dd s.
\]
Lemma~\ref{lem:robust-stability} bounds the first term by
$C\norm{z-z'}_E$.  Lemma~\ref{lem:admissible-dual-boundary}, applied on
$[0,t]$, bounds the second term by
\[
   \norm{\int_0^tU_{\widetilde K}(t,s)
      B_{\partial,N}\mathcal S_N(\xi-\eta)(s)\,\dd s}_E
      \le C\norm{\mathcal S_N(\xi-\eta)}_{L^2([0,T_c];\R^2)}
      =C|\xi-\eta|,
\]
uniformly in $t$.  This proves \eqref{eq:path-difference}.  If $\xi=\eta$,
the forced term vanishes; evaluating at $T_c$ and using
$\Phi_z^\omega=\mathcal M_N^{-1}F_z^\omega$ and the uniform bound on
$\mathcal M_N^{-1}$ proves \eqref{eq:Phi-start}.

For $i=1,2$, let $\cJ_i(t):\R^{2N}\to\R^{2N}$ be the Jacobian of
$Z_i(t)$ with respect to the Gaussian coordinates.  These Jacobians solve
\[
   \dot{\cJ}_i(t)=A_i(t)\cJ_i(t)
      +B_{\partial,N}\mathcal S_N(t),
   \qquad \cJ_i(0)=0.
\]
Here $A_i$ has the form \eqref{eq:linear-family}, with Hessian perturbation
$K_i(t)=\nabla^2W_N(q_i(t))$.  Applying
Lemma~\ref{lem:admissible-dual-boundary} to the representation
\[
   \cJ_i(t)=\int_0^tU_{K_i}(t,s)B_{\partial,N}\mathcal S_N(s)\,\dd s
\]
and using that $\mathcal S_N$ is an isometry gives
$\sup_{t\le T_c}\norm{\cJ_i(t)}_{E\to E}\le C$.  Subtracting the two
Jacobian equations yields
\[
   \cJ_1(T_c)-\cJ_2(T_c)
      =\int_0^{T_c}U_{K_1}(T_c,r)
         \bigl(A_1(r)-A_2(r)\bigr)\cJ_2(r)\,\dd r.
\]
The matrices $A_1(r)$ and $A_2(r)$ differ only in their lower-left blocks:
\[
   A_1(r)-A_2(r)
      =\begin{pmatrix}
          0&0\\
          -\bigl(K_1(r)-K_2(r)\bigr)&0
        \end{pmatrix}.
\]
Since $K_i(r)=\nabla^2W_N(q_i(r))$, applying the fundamental theorem of
calculus along the line segment from $q_2(r)$ to $q_1(r)$ and using the
definition of $\Lambda_N$ gives
\[
   \norm{K_1(r)-K_2(r)}_{\mathrm{HS}}
      \le\Lambda_N|q_1(r)-q_2(r)|.
\]
Uniform equivalence of the Euclidean and energy norms, followed by
\eqref{eq:path-difference}, therefore gives
\[
   \norm{A_1(r)-A_2(r)}_{\mathrm{HS},E}
      \le C\Lambda_N
         \bigl(\norm{z-z'}_E+|\xi-\eta|\bigr).
\]
Using the integral representation above,
$\norm{ABC}_{\mathrm{HS},E}
 \le\norm{A}_{E\to E}\,\norm{B}_{\mathrm{HS},E}\,\norm{C}_{E\to E}$,
Lemma~\ref{lem:robust-stability}, and the uniform bound on $\cJ_2$, we obtain
\[
\begin{aligned}
   \norm{\cJ_1(T_c)-\cJ_2(T_c)}_{\mathrm{HS},E}
   &\le C\Lambda_N\bigl(\norm{z-z'}_E+|\xi-\eta|\bigr)
      \int_0^{T_c}e^{-\alpha_0(T_c-r)/(2N^3)}\,\dd r\\
   &\le CN^3\Lambda_N
      \bigl(\norm{z-z'}_E+|\xi-\eta|\bigr).
\end{aligned}
\]
Finally,
\[
   D\Phi_z^\omega(\xi)-D\Phi_{z'}^\omega(\eta)
      =\mathcal M_N^{-1}\bigl(\cJ_1(T_c)-\cJ_2(T_c)\bigr).
\]
The uniform bound on $\mathcal M_N^{-1}$ now proves \eqref{eq:Dphi-HS}.
\end{proof}

We are now ready to compare the transition kernels.  The preceding estimates
verify the hypotheses of the Gaussian pushforward estimate in
Appendix~\ref{app:gaussian-pushforward} for each fixed residual path; averaging
over that path then gives the desired entropy bound.

\begin{prop}[Relative entropy between transition kernels]\label{prop:kernel-entropy}
Under Assumption~\ref{ass:weak}, estimate \eqref{eq:kernel-entropy} holds at
$T_c=\tau_cN^3$.
\end{prop}

\begin{proof}
Fix a residual path $\omega$.  For the family $\Phi_z^\omega$, with
$r(z,z')=\norm{z-z'}_E$, Lemma~\ref{lem:conditional-diffeo} gives the
derivative bound required in Lemma~\ref{lem:Gaussian-images}, while
\eqref{eq:Phi-start} and \eqref{eq:Dphi-HS} give its remaining hypotheses.  The Gaussian pushforward estimate therefore
yields
\[
   \mathsf{KL}\bigl((\Phi_{z'}^\omega)_\#\gamma_{2N}
      \mid(\Phi_z^\omega)_\#\gamma_{2N}\bigr)
      \le C\bigl[1+(N^3\Lambda_N)^2\bigr]\norm{z-z'}_E^2.
\]
Set $\mu_z^\omega=(F_z^\omega)_\#\gamma_{2N}$.  Since
$F_z^\omega=\mathcal M_N\Phi_z^\omega$ and the same invertible map
$\mathcal M_N$ is applied to both measures, relative entropy is unchanged.
Hence
\[
   \mathsf{KL}(\mu_{z'}^\omega\mid\mu_z^\omega)
      \le C\bigl[1+(N^3\Lambda_N)^2\bigr]\norm{z-z'}_E^2.
\]
The Gronwall estimate used in \eqref{eq:polygonal-solution-convergence} also
shows that $(z,\omega,\xi)\mapsto F_z^\omega(\xi)$ is jointly continuous, so
$\omega\mapsto\mu_z^\omega$ is a Borel probability kernel.

Let $\rho_N$ be the law of $W^\perp$.  By \eqref{eq:Wperp} and the independence
of $W^\perp$ and $\xi$, the transition kernel disintegrates as
\[
   \mc P_{T_c}(z,\dd y)
      =\int\mu_z^\omega(\dd y)\,\rho_N(\dd\omega).
\]
The chain rule for relative entropy and the data-processing inequality for the
projection $(\omega,y)\mapsto y$ give
\begin{align*}
   \mathsf{KL}\bigl(\mc P_{T_c}(z',\cdot)\mid
      \mc P_{T_c}(z,\cdot)\bigr)
   &\le \mathsf{KL}\bigl(
      \rho_N(\dd\omega)\mu_{z'}^\omega(\dd y)
      \mid \rho_N(\dd\omega)\mu_z^\omega(\dd y)\bigr)\\
   &=\int\mathsf{KL}(\mu_{z'}^\omega\mid\mu_z^\omega)\,
      \rho_N(\dd\omega)\\
   &\le C\bigl[1+(N^3\Lambda_N)^2\bigr]\norm{z-z'}_E^2. \qedhere
\end{align*}
\end{proof}

\begin{remark}[Conditioning and adaptedness]
The map
$(\Phi_z^\omega)^{-1}\circ\Phi_{z'}^\omega$ may depend on the entire residual
path $\omega$ and is therefore generally nonadapted.  It would be inadmissible
as an It\^o--Girsanov control.  Here it is used only in a finite-dimensional
change of variables after conditioning on $W^\perp$, where no
progressive-measurability requirement
exists.  Thus conditioning avoids the adaptedness obstruction that would arise
in a Girsanov control argument.
\end{remark}

\subsection{Entropy contraction}\label{sec:proof-stlsi}

We now combine the relative-entropy estimate for transition kernels with
contraction and a transport inequality.
For a probability measure $\mu$, write $\mu\mc P_t$ for its law after time
$t$; when $\mu=f\pi_N$, one has
$\mu\mc P_t=(\mc P_t^*f)\pi_N$.

\begin{prop}[$\mathsf{W}_E$-contraction]\label{prop:WE-contract}
Under Assumption~\ref{ass:weak}, there are $c_W,C_W>0$, depending only on
$\nu,\kappa,\gamma$, and independent of $N$, such that
\begin{equation}\label{eq:WE-contract}
   \mathsf{W}_E(\mu_1\mc P_t,\mu_2\mc P_t)
      \le C_W e^{-c_Wt/N^3}\mathsf{W}_E(\mu_1,\mu_2)
\end{equation}
for all $\mu_1,\mu_2$ with finite second moment in the energy norm.
\end{prop}

\begin{proof}
Couple the initial laws and drive both solutions by the same Brownian motion.
Their difference solves \eqref{eq:linear-family} with the mean Hessian between
the two configurations.  Lemma~\ref{lem:robust-stability} gives the pathwise
contraction.  Taking expectations and optimizing over initial couplings proves
\eqref{eq:WE-contract}.
\end{proof}

The dimension-free full-gradient LSI yields the required quadratic transport
inequality.

\begin{prop}[Uniform Talagrand inequality for the invariant measure]\label{prop:T2}
Under Assumption~\ref{ass:weak}, with
\[
   C_T:=4C_{\mathrm{LSI}}\max\{\nu+4\kappa,1\},
\]
the invariant measure satisfies the Talagrand $T_2$ inequality
\begin{equation}\label{eq:T2}
   \mathsf{W}_E(\mu,\pi_N)^2
      \le C_T \mathsf{KL}(\mu\mid\pi_N)
\end{equation}
for every probability measure $\mu\ll\pi_N$.
\end{prop}

\begin{proof}
The Gramian bound in Lemma~\ref{lem:gramian-bound} shows that
Assumption~\ref{ass:weak} implies the admissibility smallness required in
Theorem~\ref{thm:main-lsi}.  Thus the theorem applies. The
full-gradient LSI therefore gives the Euclidean Talagrand inequality.  For the homogeneous chain,
$\norm{z}_E^2\le\max\{\nu+4\kappa,1\}|z|^2$, and this comparison gives
\eqref{eq:T2}.
\end{proof}

\begin{prop}[Entropy smoothing at time $T_c$]\label{prop:smoothing}
Under Assumption~\ref{ass:weak}, every probability measure $\mu$ with finite
second moment in the energy norm satisfies
\begin{equation}\label{eq:entropy-smoothing}
   \mathsf{KL}(\mu\mc P_{T_c}\mid\pi_N)
      \le C\bigl[1+(N^3\Lambda_N)^2\bigr]
      \mathsf{W}_E(\mu,\pi_N)^2.
\end{equation}
\end{prop}

\begin{proof}
Let $\lambda\in\Pi(\pi_N,\mu)$.  Since $\pi_N\mc P_{T_c}=\pi_N$, joint
convexity of relative entropy and Proposition~\ref{prop:kernel-entropy} give
\[
\begin{aligned}
   \mathsf{KL}(\mu\mc P_{T_c}\mid\pi_N)
   &\le\int
      \mathsf{KL}\bigl(\mc P_{T_c}(y,\cdot)
         \mid\mc P_{T_c}(x,\cdot)\bigr)\,\lambda(\dd x,\dd y)\\
   &\le C\bigl[1+(N^3\Lambda_N)^2\bigr]
      \int\norm{x-y}_E^2\,\lambda(\dd x,\dd y).
\end{aligned}
\]
Optimizing over $\lambda$ proves \eqref{eq:entropy-smoothing}.
\end{proof}

\begin{proof}[Proof of Theorem~\ref{thm:main-stlsi}]
Proposition~\ref{prop:kernel-entropy} proves part (i).  We turn to the entropy
contraction needed for parts (ii) and (iii).  Write
\[
   C_{\mathsf{KL},N}=C\bigl[1+(N^3\Lambda_N)^2\bigr]
\]
for the constant in \eqref{eq:entropy-smoothing}.  Before applying that
smoothing estimate over the interval of length $T_c$, we let the dynamics
contract in Wasserstein distance for the additional time
\begin{equation}\label{eq:sN-def}
   s_N=\frac{N^3}{2c_W}
      \log_+\bigl(2C_{\mathsf{KL},N}C_W^2C_T\bigr),
   \qquad
   T_\ast=T_c+s_N,
\end{equation}
where $\log_+r=\max\{0,\log r\}$.  Let $f$ have finite entropy and put
$\mu=f\pi_N$.  Proposition~\ref{prop:T2} gives
$\mathsf{W}_E(\mu,\pi_N)<\infty$.  Since $\pi_N$ has a finite second moment by
Proposition~\ref{prop:fixedN-mixing}, so does $\mu$, and the smoothing and
contraction estimates apply.  Using them successively gives
\[
\begin{aligned}
   \Ent_{\pi_N}(\mc P_{T_c+s_N}^*f)
      &=\mathsf{KL}((\mu\mc P_{s_N})\mc P_{T_c}\mid\pi_N)\\
      &\le C_{\mathsf{KL},N}\mathsf{W}_E(\mu\mc P_{s_N},\pi_N)^2\\
      &\le C_{\mathsf{KL},N}C_W^2 e^{-2c_Ws_N/N^3}
         \mathsf{W}_E(\mu,\pi_N)^2\\
      &\le C_{\mathsf{KL},N}C_W^2C_T e^{-2c_Ws_N/N^3}
         \Ent_{\pi_N}(f)\\
      &\le\tfrac12\Ent_{\pi_N}(f).
\end{aligned}
\]
The last inequality is exactly the choice of $s_N$ in \eqref{eq:sN-def}.
Since $T_\ast=T_c+s_N$, we have proved the one-step contraction
\begin{equation}\label{eq:finite-time-entropy-proof}
   \Ent_{\pi_N}(\mc P_{T_\ast}^*f)
      \le\tfrac12\Ent_{\pi_N}(f).
\end{equation}
Because $T_c=\tau_cN^3$ and
$C_{\mathsf{KL},N}=C[1+(N^3\Lambda_N)^2]$, the definition of $s_N$ also gives
the observation-time bound \eqref{eq:TN-bound}.

Apply the integrated entropy-dissipation identity of
Lemma~\ref{lem:integrated-entropy} on $[0,T_\ast]$.  Together with
\eqref{eq:finite-time-entropy-proof}, it yields
\[
   \int_0^{T_\ast}\cI_\partial(\mc P_t^*f)\,\dd t
      =\Ent_{\pi_N}(f)-\Ent_{\pi_N}(\mc P_{T_\ast}^*f)
      \ge\tfrac12\Ent_{\pi_N}(f).
\]
This proves part (ii).

If $\sup_N N^3\Lambda_N<\infty$, then \eqref{eq:TN-bound} gives
$T_\ast\le\tau_*N^3$.  Enlarging the integration interval in
\eqref{eq:quant-stlsi} proves \eqref{eq:main-stlsi}.  For the decay estimate,
let $k=\lfloor t/T_\ast\rfloor$.  Iterating
\eqref{eq:finite-time-entropy-proof} $k$ times and using monotonicity of entropy
gives
\[
   \Ent_{\pi_N}(\mc P_t^*f)
      \le 2^{-k}\Ent_{\pi_N}(f)
      \le 2\exp\!\left(-\frac{(\log2)t}{T_\ast}\right)
      \Ent_{\pi_N}(f)
      \le 2e^{-c_*t/N^3}\Ent_{\pi_N}(f),
\]
where $c_*=(\log2)/\tau_*$.  This proves
\eqref{eq:main-entropy-decay} and part (iii).
\end{proof}

\appendix
\section{Gaussian pushforward estimate}\label{app:gaussian-pushforward}

\begin{lemma}[Relative entropy of Gaussian pushforwards]\label{lem:Gaussian-images}
Let $\gamma_d=\mathcal{N}(0,I_d)$ and let $(\Phi_x)_{x\in X}$ be $C^1$ global
diffeomorphisms of $\R^d$.  Suppose that, for a metric $r$ on $X$ and all
$x,y\in X$ and $u,v\in\R^d$,
\begin{align*}
   \norm{D\Phi_x(u)-I_d}_{\mathrm{op}}&\le\frac1{10},\\
   |\Phi_x(u)-\Phi_y(u)|&\le A r(x,y),\\
   \norm{D\Phi_x(u)-D\Phi_y(v)}_{\mathrm{HS}}
      &\le L\bigl(r(x,y)+|u-v|\bigr).
\end{align*}
Then
\begin{equation}\label{eq:Gaussian-image-entropy}
   \mathsf{KL}\bigl((\Phi_y)_\#\gamma_d\mid(\Phi_x)_\#\gamma_d\bigr)
      \le C\bigl[A^2+(1+A)^2L^2\bigr]r(x,y)^2,
\end{equation}
where $C$ is independent of $d,A,L$.
\end{lemma}

\begin{proof}
Fix $x,y\in X$, set $R=r(x,y)$, and let
\[
   \tau=\Phi_x^{-1}\circ\Phi_y,
   \qquad h(u)=\tau(u)-u.
\]
Invariance of relative entropy under $\Phi_x$ reduces the left side of
\eqref{eq:Gaussian-image-entropy} to
$\mathsf{KL}(\tau_\#\gamma_d\mid\gamma_d)$.  The first assumption and the
fundamental theorem of calculus imply that, for all $a,b\in\R^d$,
\[
   |\Phi_x(a)-\Phi_x(b)|\ge\frac9{10}|a-b|.
\]
Taking $a=\tau(u)$ and $b=u$, and using
$\Phi_x(\tau(u))=\Phi_y(u)$, we obtain
\[
   \frac9{10}|h(u)|
      \le |\Phi_y(u)-\Phi_x(u)|
      \le AR,
   \qquad\text{hence}\qquad
   |h(u)|\le\frac{10}{9}AR.
\]
Moreover,
\[
   D\tau(u)-I_d
      =D\Phi_x(\tau(u))^{-1}
         \bigl(D\Phi_y(u)-D\Phi_x(\tau(u))\bigr).
\]
The assumptions therefore give
\[
   \norm{Dh(u)}_{\mathrm{op}}\le\frac29,
   \qquad
   \norm{Dh(u)}_{\mathrm{HS}}
      \le C(1+A)LR.
\]
In particular, $\det D\tau>0$.  For $G\sim\gamma_d$, the Gaussian
change-of-variables formula yields
\[
\begin{aligned}
   \mathsf{KL}(\tau_\#\gamma_d\mid\gamma_d)
      &=\E\left[
         \frac{|\tau(G)|^2-|G|^2}{2}-\log\det D\tau(G)
      \right]\\
      &=\E\left[
         G\cdot h(G)+\frac12|h(G)|^2
         -\log\det(I_d+Dh(G))
      \right].
\end{aligned}
\]
Gaussian integration by parts gives
$\E[G\cdot h(G)]=\E[\tr Dh(G)]$.  For any matrix $H$ with
$\norm{H}_{\mathrm{op}}\le2/9$, the logarithm series and Schatten H\"older's
inequality give the dimension-free bound
\[
   \left|\tr H-\log\det(I_d+H)\right|
      \le C\norm{H}_{\mathrm{HS}}^2.
\]
Applying this estimate with $H=Dh(G)$ gives
\[
   \mathsf{KL}(\tau_\#\gamma_d\mid\gamma_d)
      \le \frac12\E|h(G)|^2+C\E\norm{Dh(G)}_{\mathrm{HS}}^2
      \le C\bigl[A^2+(1+A)^2L^2\bigr]R^2. \qedhere
\]
\end{proof}

\bibliographystyle{amsxport}
\bibliography{refs}

\end{document}